\documentclass[10pt,journal,compsoc]{IEEEtran}
\usepackage{times}
\usepackage{amsmath}
\usepackage{float}
\usepackage{graphicx}
\usepackage{amsthm}
\usepackage{url}
\usepackage{xcolor,graphicx,float}
\usepackage{subfigure}
\usepackage{caption}
\usepackage{color}
\usepackage{multirow}
\usepackage{epstopdf}
\usepackage{color}
\newtheorem{theorem}{Theorem}
\newtheorem{definition}{Definition}
\newtheorem{mechanism}{Mechanism}
\newtheorem{lemma}{Lemma}
\newcommand{\E}{{\rm I\kern-.3em E}}
\newcommand{\Var}{\mathrm{Var}}

\begin{document}

\title{Strategic Learning Approach for Deploying UAV-provided Wireless Services 
}

\author{Xinping~Xu,
        Lingjie~Duan,~\IEEEmembership{Senior~Member,~IEEE,}
        and~Minming~Li,~\IEEEmembership{Senior~Member,~IEEE}
\thanks{\quad\em L.~Duan and X.~Xu were supported by the Singapore Ministry of Education Academic Research Fund Tier 2 under Grant MOE2016-T2-1-173.
M.~Li was partially supported by NNSF of China under Grant No. 11771365, and  by Project No. CityU 11200518 from Research Grants Council of HKSAR.
Part of work has been presented in the 16th IEEE International Symposium on Modeling and Optimization in Mobile, Ad Hoc and Wireless Networks (WiOpt 2018), Shanghai, China, May 7-11, 2018 \cite{xu2018uav}. (\emph{Corresponding author: X.~Xu.})
}
\thanks{\quad\em X.~Xu and L.~Duan are with Engineering Systems and Design Pillar, Singapore University of Technology and Design, Singapore.
E-mail:{xinping\_xu@mymail.sutd.edu.sg}; {lingjie\_duan@sutd.edu.sg}}.
\thanks{\em M.~Li is with Department of Computer Science, City University of Hong Kong, Hong Kong SAR. Email: minming.li@cityu.edu.hk}
}

\IEEEtitleabstractindextext{%
\begin{abstract}
Unmanned Aerial Vehicle (UAV) have
emerged as a promising technique to rapidly provide wireless
services to a group of mobile users simultaneously.
		The paper aims to address a challenging issue that each user is selfish and may misreport his location or preference for changing the optimal UAV location to be close to himself.
		Using algorithmic game theory, we study how to determine the final location of a UAV in the 3D space, by ensuring all selfish users' truthfulness in reporting their locations for learning purpose.
		To minimize the social service cost in this UAV placement game, we design strategyproof mechanisms with the approximation ratios, when comparing to the social optimum.
		We also study the obnoxious UAV placement game to maximally keep their social utility, where each incumbent user may misreport his location to keep the UAV away from him.
		Moreover, we present the dual-preference UAV placement game by considering the coexistence of the two groups of users above, where users can misreport both their locations and preference types (favorable or obnoxious) towards the UAV.
Finally, we  extend the three games above to include multiple UAVs and design strategyproof mechanisms with provable approximation ratios.
\end{abstract}
}

\maketitle
\IEEEpeerreviewmaketitle



\section{Introduction}\label{section_introduction}

\IEEEPARstart{F}{uture} development of unmanned aerial vehicles (UAVs) expects each UAV to be intelligent enough to learn
 and operate independently  without intervention  of human controllers.
As a hot topic of artificial intelligence, algorithmic game theory helps design such a UAV to strategically interact with its potential customers in the first place and learn their private information before flying to provide customized services \cite{ernest2016genetic}.
Recently, in the field of wireless communications, the use of UAVs as flying
cell sites becomes a promising technique to solve the coverage problem
of territorial wireless networks \cite{mozaffari2016efficient}.
Traditional base stations are deployed at fixed locations on the ground for a long term
by catering to the average traffic load in the two-dimensional ground area, while flying UAVs do not have such constraint in
space or time for deployment.
Owing to their agility and mobility, UAVs can be quickly deployed as alternatives to meet time-varying traffic load.
Major wireless carriers such as AT\&T started to use UAVs to opportunistically boost wireless coverage for crowds in big concerts and sports, where people request good wireless services to continuously post their selfies and videos online \cite{drones-web}.
Moreover, UAVs can be rapidly deployed in events of disasters to enable air-to-ground communications when territorial base stations fail to work.
Verizon launched an exercise in mid 2017 to deploy UAVs to Cape May, New Jersey, and provide emergency crews facing hurricane disaster with airborne 4G  LTE connectivity \cite{pressman}.

To fully reap
the benefits of UAV-provided wireless services,
one must determine the final UAV hovering location to best serve a group of target users in the two-dimensional (2D) geographical ground.
As the number of UAVs is small compared to the target user size, the final UAV position needs to balance all target users' different locations and preferences to optimize their social benefit.
Such a problem has been recently investigated in the literature by assuming that the UAV knows the real locations or at least the distribution of mobile users upon deployment (e.g., \cite{6863654,7888557,zhang2018fast,pan2012cooperative,wang2019dynamic}).
For example, \cite{6863654} aimed to maximize the UAV's wireless coverage on the ground, by considering the air-to-ground signal propagation feature.
\cite{7888557} improved the energy-efficiency of UAV communication with ground users by designing the UAV's trajectory.
\cite{zhang2018fast} studied how to minimize the delay of deploying UAVs till providing the full wireless coverage in the worst scenario.
\cite{pan2012cooperative} studied the cooperation among intelligent vehicles in term of link scheduling.
\cite{wang2019dynamic} further studied how a UAV should allocate and price its limited capacity for serving a group of users on the way, by assuming public user arrival and preference distributions.
Unlike these works, we aim to study the optimal UAV placement without knowing any user's location or distribution information beforehand, by strategically learning from the selfish users themselves.

In practice, it is difficult for a UAV to track users' locations in real time, and traditional user positioning techniques require multiple base stations' continuous help \cite{901174,gu2009survey}.
When requiring a UAV's help, however, ground network infrastructure is often congested and may even fail to work, which makes it difficult for a UAV to self-track users' locations upon deployment  \cite{kos2006mobile}.
It is more desirable for the UAV to directly interact with users for learning their own locations  for the optimal deployment.
Though appealing, this approach is vulnerable
given many users in practice are selfish and may misreport their private information.
Selfish users only care about their own service benefits and prefer the closest UAV location to themselves.
They may  not report their true locations to help determine the optimal UAV placement for best serving the crowd.
Consider an illustrative uplink communication example where we deploy a UAV to a hovering point on a line interval $[0,4]$ for serving user $1$ at location $x_1=0$ and user $2$ at location $x_2=2$ simultaneously.
Each user prefers the final UAV location to be as close to his own location as possible to obtain high signal-to-noise-ratio or save his transmission power.
If the two users report their locations truthfully, the UAV chooses to locate at the mean of the two users' locations (i.e., optimal UAV location $x=1$).
However, if user $2$ misreports his location from $x_2=2$ to  $x_2'=4$, then mean UAV location changes to $x=2$ which is the closest to user $2$.
For optimizing the wireless service provision, in this paper, we use algorithmic game theory to investigate how the UAV should interact strategically with selfish users and learn their true locations and preferences.

Besides the favorable UAV placement game, we also study the obnoxious UAV placement game.
As the new UAV facility may interfere with another group of incumbent (adverse) users in the same space, we want to best control the interference and maximally keep these users' social utility when the UAV determines its placement position.
In this game, the UAV also requires all such users to report their locations for determining the UAV location, where a user may misreport his location to mislead the final UAV location to be further away from his true location and reduce interference from the UAV.
Moreover, as both the UAV's facility users (who prefer to be close to the UAV) and adverse users (who prefer to be far away from the UAV) may coexist at the same time, we want to reach a good balance between the positive and negative effects on the two diverse user groups when locating the UAV in the dual-preference UAV placement game. We still require strategyproof (truthful) mechanism design for the strategic UAV's learning interaction with selfish users to ensure all users' truthfulness in reporting their locations and even preference types.


The paper's key novelty and main contributions are summarized as
follows.

\begin{itemize}
\item \emph{Novel UAV placement  games for strategic learning of users' private information:}
To our best knowledge, our paper is the first to propose and analyze UAV placement games through strategically learning selfish users' locations and service preferences.
We consider a challenging scenario where mobile users can purposely hide their current locations and preferences from the UAV.
Using algorithmic game theory, we completely study a UAV placement game for serving facility users, an obnoxious UAV placement game for protecting adverse users, and a dual-preference UAV placement game for handling both groups of users, where users are selfish and may misreport their locations to mislead the UAV placement. We aim to design strategyproof mechanisms with provable approximation ratios in these three games to ensure users' truthful location reporting and optimize the social cost or utility.

\item \emph{Mechanism design for the UAV placement game:}
In Section \ref{sectionflg}, we propose two strategyproof mechanisms such that any user's misreporting of his locations can only increase his service cost.
Especially, we design the weighted median strategyproof mechanism for the strategic UAV-user learning interaction. with approximation ratio $2^{(3\alpha-4)/2}$, when comparing to the social optimal cost.
This worst-case result is robust no matter which
distribution users' locations follow.
Besides the worst-case analysis of the proposed strategyproof mechanisms, we also analyze the empirical performances of the two mechanisms and prove that they converge to the social optimum as the number of users becomes large, given users' locations following any symmetric distributions.

\item \emph{Mechanism design for the obnoxious UAV placement game:}
In Section \ref{sectionoflg}, we consider the opposite problem of locating an obnoxious UAV.
Each (adverse) user now attempts to stay far away from the UAV to reduce its received interference, by misreporting his location. Our target is to design a strategyproof mechanism for the UAV placement  and maximally keep the social utility of such users. Accordingly, we design a strategyproof mechanism with approximation ratio $5\times 2^{(\alpha-2)/2}$. Our empirical analysis further shows that the proposed mechanism converges to the social optimum as the number of users becomes large, given users' locations following asymmetric distributions.

\item \emph{Mechanism design for the dual-preference UAV placement game:}
In Section \ref{sectiondpflg}, we study the more general case of the dual-preference UAV placement game by considering the co-existence of both facility users and adverse users.
Besides locations, we further allow users to misreport their preference types (i.e., favorable or obnoxious) in the strategic UAV-user interaction.
We design a strategyproof mechanism with approximation ratio $2^{3\alpha/2}$ and validate the empirical result that it converges to the social optimum as the number of users becomes large, given users' locations following any asymmetric distributions
and the number of adverse users is larger than the number of facility users.

\item \emph{Mechanism design for multi-UAV placement game:}
In Section \ref{section_multiple}, we extend our three UAV placement games  in Sections \ref{sectionflg}-\ref{sectiondpflg} by including more than one UAV.
In the dual-preference placement game with two UAVs, we further allow each user to have diverse and hidden preferences over different UAVs and design a strategyproof  mechanism with approximation ratio
When there is an arbitrary number $k$ of UAVs to deploy, we also extend our mechanism design in the favorable and obnoxious placement games.
\end{itemize}



The organization of the paper is shown as follows.
In Section \ref{section_system}, we describe the mathematical models of one-UAV placement and obnoxious one-UAV placement game;
in Section \ref{sectionflg} of one-UAV placement game, we propose strategyproof Mechanisms \ref{m1} and \ref{m2} and analyze the empirical performances of them;
in Section \ref{sectionoflg} of obnoxious one-UAV placement game, we propose strategyproof Mechanisms \ref{m4} and \ref{m1.1}; 
in Section \ref{sectiondpflg} of dual-preference UAV placement game, we propose strategyproof Mechanism \ref{m7}; 
in Section \ref{section_multiple} of multiple-UAV placement game, we propose strategyproof Mechanisms \ref{m110} and \ref{m10}.

\subsection{Related work}
In the research of algorithmic game theory, there are some studies on the generic facility location game and strategyproof mechanisms to prevent users from misreporting locations.
Such mechanisms are simply based on users' location reports and are easy to implement, without using complicated schemes such as location-based pricing and billing.
For example, \cite{ProcacciaandTennenholtz} proposed median strategyproof mechanisms with provable approximation ratios on a one-dimensional line, which gives us some inspiration of proposing our strategyproof mechanisms in the UAV placement game.
In the obnoxious facility location game, the mechanism design for
the objective of maximizing total users' utility was first studied by \cite{cheng2011mechanisms}. \cite{ibara2012characterizing} characterized strategyproof mechanisms with
exactly two candidates in the general metric and showed that there exists
a lower bound $3$ of strategyproof mechanism.
\cite{zou2015facility} and \cite{feigenbaum2015strategyproof} investigated the properties of the one facility location game with dual-preference.
\cite{sui2013analysis} showed mechanism design problem for users over multi-dimensional domains when multiple facilities can be chosen.


Such works focus on one facility placement in one-dimension, while the real UAV placement is in 3D
and may include more than one UAV facility.
Further, our paper models users' heterogeneity in their service sensitivities (weights) and the line-of-sight air-to-ground propagation in wireless communications.
Such unique wireless feature and user heterogeneity translate to a new problem objective and require
new methods in designing the strategyproof
mechanisms and proving approximation ratios.

\section{System Model and Problem Formulation}\label{section_system}

\begin{figure}[t]
	\centerline{\includegraphics[width=8.5cm]{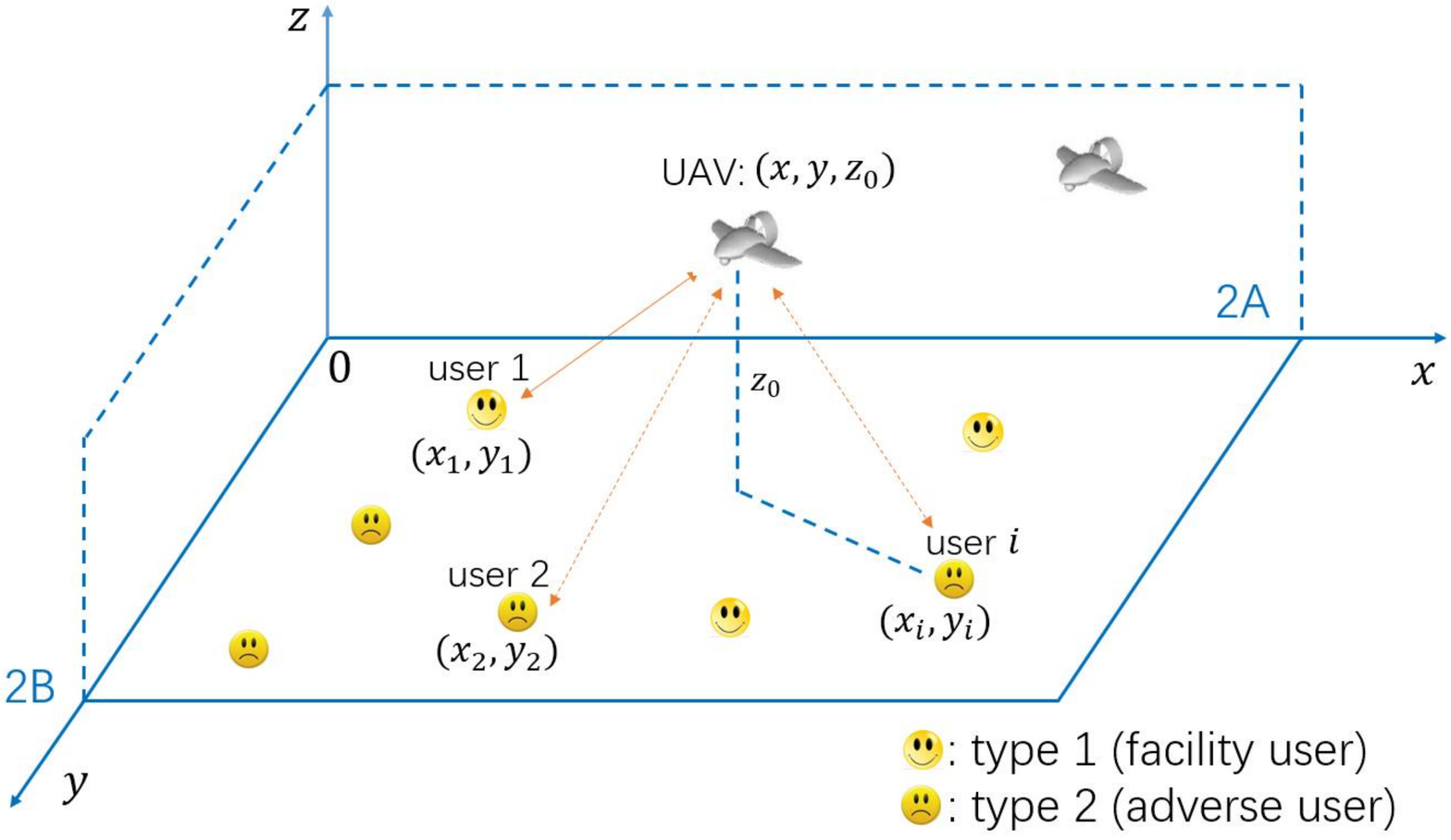}}
		\caption{System model about a UAV's placement in the 3D space. There are generally two types of users with dual preferences: type 1 (users of the UAV facility) and type 2 (adverse users experiencing UAV's interference). In the UAV placement game in Section \ref{sectionflg}, all users are of type 1;  in the obnoxious UAV placement  game in Section \ref{sectionoflg}, all users are of type 2; and in the dual-preference UAV placement game in Section \ref{sectiondpflg}, both types of users coexist.}
\label{fig1}
\end{figure}

Let $N = \{1,2,\cdots,n\}$ be the set of users that are located in the 2D space $I^2$. Without loss of generality, we suppose  $I^2$ is a finite rectangle $[0,2A]\times [0,2B]$ containing all $n$ users as shown in Fig.\ref{fig1}. The real location of user $i\in N$ is $ (x_i,y_i)\in I^2$. We denote $ \textbf x =(x_1,x_2,\cdots,x_n )$ and $\textbf y =(y_1,y_2,\cdots,y_n )$ as users' location profiles in the 2D space.
Depending on the users' locations, the single UAV's location is denoted as point $(x, y, z_0),$ which is in 3D space. \footnote{We will extend all solutions to include more than one UAV in Section \ref{section_multiple}.}
$z_0$ satisfies $z_0\in [0,\infty)$ and is a fixed number determined by the UAV.
The distance between user $i$ and the UAV is $\sqrt{(x_i-x)^2+(y_i-y)^2+z_0^2} $.

We first introduce the UAV placement game, where each user (of type 1 in Fig.\ref{fig1}) prefers the UAV location to be close to his own location for saving his service cost.
Similar to \cite{wangzhe2018},
we model air-to-ground link for the suburban scenarios by using Rician fading with $\kappa$ factor, where $\kappa$ is the ratio between the energy in the line-of-sight (LoS) component and the energy in the multi-path component.
Note that Rician fading is an adequate choice which consists of
an LoS component and a large number of i.i.d. reflected and scattered waves.
Similar to \cite{6863654}, we also fix UAV altitude as $z_0$ to best trade off between LOS links and signal attenuation.
Then we generally model the coordinate of the UAV in 3D as $(x,y,z_0)$ and the coordinate of ground user $i$ in 2D as  $(x_i,y_i)$, respectively.

We thus model the
instantaneous communication channel between the UAV and each user $i$ as a product of a
large-scale path loss component and a small-scale fading component, i.e.,
\begin{align}
g_i=|\psi_i|^2\theta[\sqrt{(x-x_i)^2+(y-y_i)^2+z_0^2}/d]^{-\alpha},\nonumber
\end{align}
where  $\psi_i$ is the small-scale fading gain of the channel between the UAV and user $i$, $\alpha\in[2,6]$ is the path-loss exponent, and $\theta(dB)=-20\log_{10}(4\pi d/\nu)$ denotes the channel power at the
reference distance of $d$ with wavelength $\nu.$
We adopt $d = 1$ meter throughout the paper
and assume the additive white Gaussian noise has zero mean and variance $\sigma^2.$
We consider the small-scale fading gain $\psi_i$ follows Rician fading $\psi_i=\sqrt{\frac{\kappa}{1+\kappa}}\psi_L+\sqrt{\frac{1}{1+\kappa}}\psi_M$ \cite{paulraj2003introduction}, where $\kappa$ is the ratio between the energy in the LoS component and the energy in the multi-path component,
$\psi_L$ is a normalized constant representing the LoS component and
$\psi_M$ is the circular symmetric complex Guassian random variable with zero mean and unit variance.
Denote the transmit power by user $i$ as $c_i$.
The instantaneous signal-to-noise ratio (SNR) for user $i$'s signal at the UAV receiver is given by
\begin{align}
{\mathtt{SNR_i}}=\frac{c_i g_i}{\sigma^2}=\frac{c_i |\psi_i|^2 \theta}{((x-x_i)^2+(y-y_i)^2+z_0^2)^{\alpha/2}\sigma^2}.\nonumber
\end{align}
After determining the location of the UAV and user $i$, the instantaneous
success probability under small-scale fading is given by $\Pr(\mathtt{SNR_i}\geq \mathtt{SNR_{i,th}}),$ where $\mathtt{SNR_{i,th}}$ is the instantaneous SNR threshold for supporting user $i$'s application. To ensure this probability is sufficiently large, we have the condition of $\Pr(\mathtt{SNR_i}\geq \mathtt{SNR_{i,th}})\geq 1-\epsilon$ and rewrite it as $\mathtt{\overline{SNR}_i}\geq \mathtt{\overline{SNR}_{i,th}}$, where $\mathtt{\overline{SNR}_{i,th}}$ is a function of $\mathtt{SNR_{i,th}}$ and $\epsilon$,
and $\mathtt{\overline{SNR}_i}$ is the average SNR give by,
\begin{align}
{\mathtt{\overline{SNR}_i}}&=\frac{c_i \E[g_i]}{\sigma^2}=\frac{c_i \E[|\psi_i|^2] \theta}{((x-x_i)^2+(y-y_i)^2+z_0^2)^{\alpha/2}\sigma^2}.\nonumber\\
&=\frac{c_i \theta}{((x-x_i)^2+(y-y_i)^2+z_0^2)^{\alpha/2}\sigma^2}\geq \mathtt{\overline{SNR}_{i,th}}.\nonumber
\end{align}
For ease of reading, we denote the weight of user $i$ as $w_i=\mathtt{\overline{SNR}_{i,th}}\sigma^2/\theta,$ and expect the minimum
the service cost
\begin{align}
c_i&=\mathtt{\overline{SNR}_{i,th}}\sigma^2 ((x-x_i)^2+(y-y_i)^2+z_0^2)^{\alpha/2}/\theta\nonumber\\
&=w_i  ((x-x_i)^2+(y-y_i)^2+z_0^2)^{\alpha/2}\nonumber
\end{align}
in term of power consumption.
Here,
weight $w_i>0$ models user $i$'s sensitivity 
and he prefers the UAV to be closely located for his cost saving.
If a user has a larger weight as compared to the other users, the UAV should be located closer to him.
We denote $w_i$ as user $i$'s weight and $ \textbf w =(w_1,w_2,\cdots,w_n )$  as weight profile.
In practice, each user $i$'s weight $w_i$ can be estimated by the UAV from identifying the specific traffic application type, and is public information.
\footnote{Actually, some of our mechanisms (e.g., Mechanisms \ref{m1} and \ref{m1.1}) designed later also handles the case without knowing $w_i.$}
However, the UAV does not know the users' location profiles $\textbf{x}$ and $\textbf{y}$.
We denote $\Omega=\{\textbf x,\textbf y|\textbf w\}$ as the full user profile.
The UAV's objective is to minimize the sum of weighted costs by 
 choosing $(x,y,z_0)$ based on
users' locations reports.

In the UAV placement game, a mechanism outputs a UAV location $(x, y, z_0)$ based on a given profile $\Omega$ and thus is a function $f:I^{2n} \to I^2$, i.e., $(x, y)=f(\textbf x,\textbf y).$
As explained, the cost of user $i$ is given by
\begin{align}
	c_i(f(\textbf x,\textbf y),(x_i, y_i))
= &w_i ((x_i\!-x)^2+(y_i\!-y)^2+z_0^2)^{\alpha/2}.
	\label{a1}
\end{align}
Let $\textbf x_{-i}=(x_1,\cdots, x_{i-1}, x_{i+1},\cdots, x_n)$ and $\textbf y_{-i}=(y_1,\cdots,$
$y_{i-1}, y_{i+1},\cdots, y_n)$ denote the location profiles for all $n$ users except user $i$. 
The social cost of a mechanism $f$  is defined as the sum of all users' costs, i.e.,
\begin{align}
	SC(f(\textbf x,\textbf y),(\textbf x,\textbf y))&=\sum_{i=1}^n c_i(f(\textbf x,\textbf y),(x_i,y_i)).
	\label{a3}
\end{align}

Each user $i$ will finally reach the
target average SINR ($\mathtt{\overline{SNR}_{i,th}}$ for user $i$) regardless of the UAV location after deployment.
Thus, the SINR
information is a constant for each user $i$.
Since user $i$ can flexibly adjust its transmit power $c_i,$ it can
always meet the target average SINR ($\mathtt{\overline{SNR}_{i,th}}$) yet may incur a large transmit power cost $c_i.$
To minimize the total transmit power cost $\sum_{i=1}^n c_i$ in the above problem's objective, we further need to know the users¡¯ location information. Thus, we propose users' location reporting of $(x_i, y_i)$, and design truthful mechanisms to collect such reliable location information.
Only after obtaining users' location information before deployment, the UAV can estimate the total objective of problem above at various UAV location $(x, y, z_0)$, and compare to choose the best UAV location.
In the following, we formally define the strategyproofness for mechanism design in the UAV placement  game, which is robust against any distributions of users' locations.
\begin{definition}\label{d1}
	A mechanism is strategyproof in the UAV placement  game if no user can benefit from misreporting his location. Formally, given profile $\Omega=(<x_i,\textbf x_{-i}>, <y_i,\textbf y_{-i}> |\textbf w)\in I^{2n}$, and any misreported location $(x_i'  ,y_i')\in I^2$ for any user $i\in N$, it holds that
\begin{align}
 &c_i(f((x_i,y_i),(\textbf x_{-i},\textbf y_{-i})),(x_i, y_i))\nonumber\\
 \leq & c_i(f((x_i'  ,y_i'),(\textbf x_{-i},\textbf y_{-i})),(x_i, y_i)). \nonumber
\end{align}
\end{definition}

For the UAV placement  game, we are interested in designing strategyproof mechanisms that perform well with respect to minimizing the social cost.  Given a location profile $\Omega$, let $OPT_1(\textbf x,\textbf y)$ be the optimal social cost in (\ref{a3}). A strategyproof mechanism $f$ has an approximation ratio $\gamma\geq 1$, if for any location profile $(\textbf x,\textbf y)\in I^{2n},$
$\gamma OPT_1(\textbf x,\textbf y)\geq SC(f,(\textbf x,\textbf y))$.
$\gamma$ tells us the worst-case performance of $f$ no matter which distributions the users' locations follow, and we prefer $f$ with a small $\gamma$.

On the other hand, in the obnoxious UAV placement game,
the UAV faces a different group of $n$ adverse users (of type 2 in Fig.\ref{fig1}) and introduces downlink interference to them. They
prefer to be far away from the UAV and their (positive) weights $w_i$'s here tell their
different interference sensitivities in their traffic applications. We define
adverse user $i$'s utility $u_i=w_id((x_i, y_i), (x,y,z_0))^2$ under interference, which is the same as (\ref{a1}).
$u_i$ nonlinearly increases with the distance from the UAV. Opposite to the UAV placement game, the UAV's objective in this game is to maximize the sum of users' weighted
utilities, by designing strategyproof mechanisms $f(x,y,z_0)$ for learning users' truthful locations.
The social utility of a mechanism $f$ is defined as:
\begin{align}
SU(f(\textbf x,\textbf y),(\textbf x,\textbf y))&=\sum_{i=1}^n u_i(f(\textbf x,\textbf y),(x_i,y_i)).
\label{a15}
\end{align}

Next, we formally define the strategyproofness for the obnoxious UAV placement game.
	
\begin{definition}\label{d2}
		A mechanism is strategyproof in the obnoxious UAV placement game if no adverse  user can benefit from misreporting his location. Formally, given profile $\Omega=(<x_i,\textbf x_{-i}>, <y_i,\textbf y_{-i}, y_i>|\textbf w)\in I^{2n},$ and any misreported location $(x_i'  ,y_i')\in I^2$ for user $i$, it holds that
\begin{align}
 & u_i(f((x_i,y_i),({\textbf x_{-i}},\textbf y_{-i})),(x_i,y_i))\nonumber\\
 \geq & u_i(f((x_i'  ,y_i'),(\textbf x_{-i},\textbf y_{-i})),(x_i,y_i)). \nonumber
\end{align}
\end{definition}

For the obnoxious UAV placement game, we are interested in designing strategyproof mechanisms that perform well with respect to maximizing the social utility in (\ref{a15}). Given a location profile $\Omega$, let $OPT_2(\textbf x,\textbf y)$ be the optimal social utility. A strategyproof mechanism $f$ has an approximation ratio $\gamma\geq 1$, if for any location profile $(\textbf x,\textbf y)\in I^{2n}, OPT_2(\textbf x,\textbf y)\leq \gamma SU(f,(\textbf x,\textbf y))$.

We will introduce the model of the dual-preference UAV placement game in Section \ref{sectiondpflg}, by combining the two games as defined above. In Section \ref{section_multiple}, we will further extend the three multi-UAV placement games above by including multiple UAVs for new strategyproof mechanisms.

\section{UAV Placement Game for Type 1 Users }\label{sectionflg}

In this section, we design strategyproof mechanisms for the UAV placement game where all $n$ users are of preference type 1. According to (\ref{a1}) and (\ref{a3}), we have the following social cost
\begin{align}
	 SC(f,(\textbf x,\textbf y))=\sum_{i=1}^n w_i ((x_i-x)^2+(y_i-y)^2+z_0^2)^{\alpha/2},\nonumber
\end{align}
which is a convex function with respect to $(x,y)$.
The optimization problem of this game  is formulated as
\begin{align}
\begin{cases}
	&\min_{f}  \sum_{i=1}^n w_i ((x_i-x)^2+(y_i-y)^2+z_0^2)^{\alpha/2},\\
& \mbox{s.t.
for any misreported location $(x_i'  ,y_i')\!\in\! I^2$ and user $i\!\in\! N$,}\\
&\quad c_i(f((x_i,y_i),(\textbf x_{-i},\textbf y_{-i})),(x_i, y_i))\nonumber\\
&\leq  c_i(f((x_i'  ,y_i'),(\textbf x_{-i},\textbf y_{-i})),(x_i, y_i)).\\
&\mbox{Variable: function } f(\textbf x, \textbf y)=(x,y): I^{2n} \to I^2.\\
&\mbox{Parameters:} (x_i, y_i)\in I^2,w_i>0, \mbox{for } i=1,\dots,n,
 z_0\geq 0  \\
 &\mbox{and }\alpha\in [2,6].
 \end{cases}
\end{align}
For $\alpha=2$, by checking the first-order conditions,  we obtain the weighted mean $(x, y, z_0)=(\bar{x},\bar{y}, z_0)$ as the optimal location, where
\begin{align}
\bar{x}=\frac{\sum_{i=1}^n w_i x_i}{\sum_{i=1}^n w_i}
\mbox{ and }
\bar{y}=\frac{\sum_{i=1}^n w_i y_i}{\sum_{i=1}^n w_i}.
\label{a17}
\end{align}
However, this weighted mean mechanism is not strategyproof as we explained in the illustrative example in Section \ref{section_introduction}.
\subsection{Design and Analysis of strategyproof mechanisms}
In the following, we present two strategyproof mechanisms with provable approximation ratios.

\begin{mechanism}\label{m1}
		Given a profile $\Omega$, return median location $(x,y,z_0)=med(\textbf x,\textbf y, z_0)=(x_{med}, y_{med}, z_0)$ as the UAV location, where $x_{med}$ is the median of $\textbf x$,
\footnote{If $n$ is even, we choose the $(\frac{n}{2})$-th smallest  value of $\textbf x$ profile as $x_{med}$. This location strategy is the same for location profiles $\textbf y$ and $\textbf z$.}
and $y_{med}$ is the median of $\textbf y$.

\end{mechanism}

\begin{theorem}\label{t1}
	Define $w_{max}=\max¡\{w_1,\dots,w_n\}$ and $w_{min}=\min ¡\{w_1,\dots,w_n\}.$ Mechanism \ref{m1} is strategyproof for $\alpha\in[2,6].$ and has the approximation  ratio $\gamma=\frac{w_{max}}{w_{min}} 2^{(3\alpha-4)/2}$ for $\alpha\geq2$ as compared to the social optimum.
\end{theorem}

\begin{proof}
First we prove Mechanism \ref{m1} is a strategyproof mechanism.
Assume $x_1\leq x_2\dots\leq x_n$ without loss of generality, and  $x$-location of UAV $x_{med}$ is $x_j$ (i.e., $x=x_{med}=x_j$).
If user $i$ ($i\leq j$) chooses to misreport his $x$ location, we have two cases:
(\romannumeral1) The misreported $x$-value is smaller than the original $x$-value $x_j$ and the $x$-value of the new UAV location (i.e., $x$) will not change;
(\romannumeral2) The misreported $x$-value is greater than the original $x$-value $x_j$ and the $x$-value of the new UAV location (i.e., $x$) will not be smaller than $x_j$. However, $(x_i-x)^2$ will not decrease and thus his cost $w_i ((x_i-x)^2+(y_i-y)^2+z_0^2)^{\alpha/2}$ will not decrease.
Therefore, user $i$ cannot decrease his cost by misreporting his $x_i$. Similarly, he cannot decrease his cost by misreporting his $y_i$ in the other independent domain of the 2D space.
Similar results hold for $i>j$ due to symmetry.
Next, we prove the approximation ratio $\gamma.$

To obtain the approximation ratio, we need to let $z_0=0$ first.
By Lemmas 1 and 2 which are  given in Appendix \ref{app_1} and Appendix \ref{app_2}, respectively,
we have
\begin{align}\label{a52}
\gamma&
=\frac{\sum_{i=1}^n w_i ((x_i-x_{med})^2+(y_i-y_{med})^2)^{\alpha/2}}{\min_{x,y}\sum_{i=1}^n w_i ((x_i-x)^2+(y_i-y)^2)^{\alpha/2}}\nonumber\\
&\leq \frac{w_{max}}{w_{min}} \frac{\sum_{i=1}^n ((x_i-x_{med})^2+(y_i-y_{med})^2)^{\alpha/2}}{\min_{x,y}\sum_{i=1}^n  ((x_i-x)^2+(y_i-y)^2)^{\alpha/2}}\nonumber\\
&\leq \frac{w_{max}}{w_{min}} \frac{\sum_{i=1}^n 2^{\alpha/2-1}((x_i-x_{med})^\alpha+(y_i-y_{med})^\alpha)}{\min_{x,y}\sum_{i=1}^n  ((x_i-x)^\alpha+(y_i-y)^\alpha)}\nonumber\\
&=\frac{w_{max}}{w_{min}} \frac{2^{\alpha/2-1}\sum_{i=1}^n (x_i-x_{med})^\alpha}{\min_{x}\sum_{i=1}^n  (x_i-x)^\alpha}\nonumber\\
&\leq \frac{w_{max}}{w_{min}} 2^{\alpha/2-1} 2^{\alpha-1}=\frac{w_{max}}{w_{min}} 2^{(3\alpha-4)/2}.
\end{align}
\end{proof}
Mechanism \ref{m1} treats each user equally and does not consider users' weights.  If users have diverse weights such that $\frac{w_{max}}{w_{min}}$ is large, the approximation ratio $\gamma$ is large.
It should be noted that Mechanism \ref{m1} also has its merit:  since the UAV does not need to gather the information of weights from users, it is strategyproof even if we allow users to misreport their weights. Next, we propose a better mechanism to achieve a much smaller approximation ratio.

\begin{mechanism}\label{m2}
	
Consider $x$-domain first, we reorder $\{x_1,x_2,\dots,x_n\}$ as $\{x_{j_1},x_{j_2},\dots,x_{j_n}\}$ with $x_{j_1}\leq x_{j_2}\leq \dots\leq x_{j_n}.$
Define $x_{wmed}$ as a particular $x_{j_{q}}$, where integer $q$ satisfies  $\sum_{i\leq q}w_{j_{i}}\geq\sum_{i>q}w_{j_{i}}$ and $\sum_{i< q}w_{j_{i}}<\sum_{i\geq q}w_{j_{i}}.$
In $y$-domain, $y_{wmed}$ follows the same structure.
	Given a profile $\Omega$, the UAV adopts weighted median $wmed(\textbf x,\textbf y,z_0)=(x_{wmed}, y_{wmed},z_0)$ for its location.
\end{mechanism}

\begin{theorem}\label{t2}
	Mechanism \ref{m2} is strategyproof for $\alpha\in[2,6]$ and has the approximation ratio $\gamma=2^{(3\alpha-4)/2}$ for $\alpha\geq 2.$
\end{theorem}

\begin{proof}
First, we can use the similar analysis in Theorem \ref{t1} to prove that Mechanism \ref{m2} is strategyproof for $\alpha\in[2,6]$.
Now we prove the approximation ratio for $d\geq2$.
Without loss of generality, we rescale  each  $w_i$ uniformly as positive integer in this proof.
By partitioning user $i$ into a number $w_i$ of small users with unit weight $1$,
we obtain new sequenced sets of profile $\textbf{x}$ and $\textbf{y}:$
\begin{align}
&\{\underbrace{x_{j_1},\dots,x_{j_1}}_{w_{j_1}},\underbrace{x_{j_2},\dots,x_{j_2}}_{w_{j_2}},\dots,\underbrace{x_{j_n},\dots,x_{j_n}}_{w_{j_n}}\},\nonumber\\
&\{\underbrace{y_{j_1},\dots,y_{j_1}}_{w_{j_1}},\underbrace{y_{j_2},\dots,y_{j_2}}_{w_{j_2}},\dots,\underbrace{y_{j_n},\dots,y_{j_n}}_{w_{j_n}}\}.
\label{a7}
\end{align}
Then we rewrite $SC(f,(\textbf x,\textbf y))\leq \gamma OPT_1(\textbf x,\textbf y)$ as
\begin{align}
&\sum_{i=1}^n(\underbrace{{(x_{j_i}-x_{wmed})}^2+{(y_{j_i}-y_{wmed})}^2)^{\alpha/2}+\dots}_{w_{j_i}})\nonumber\\
\leq &\gamma \min_{x,y}\sum_{i=1}^n(\underbrace{{(x_{j_i}-x)}^2+{(y_{j_i}-y)}^2)^{\alpha/2}+\dots}_{w_{j_i}})
\label{a16}
\end{align}
Note that $x_{wmed}$ is the median in set (\ref{a7}) and
$y_{wmed}$ is the median in set (\ref{a7}).
According to the similar proof in (\ref{a52}), we can prove  in (\ref{a16}), $\gamma=2^{(3\alpha-4)/2}.$

\end{proof}
Comparing Mechanisms \ref{m1} and \ref{m2}, we can see that Mechanism \ref{m2} takes users' heterogeneous weights into account for the mechanism design and achieves better worst-case performance. Besides the worst-case analysis, we will show in Subsection \ref{subsection_empirical} that these two mechanisms perform analogously in average sense to approach the social optimum.
The complexity of the proposed two mechanisms should be $O(n).$
Such complexity is low and can return solution quickly.
The computation complexities of the following Mechanisms 3-5 are linear in
the number $n$ of users and thus in the order also are $O(n).$


\subsection{Empirical analysis of Mechanisms \ref{m1} and \ref{m2}}\label{subsection_empirical}

So far we have only analyzed the worst-case performances of the two mechanisms in term of approximation ratio. In this subsection, we present empirical analysis to further evaluate the average performances of the mechanisms for $\alpha=2$.
In Mechanism \ref{m1}, we choose median location as the UAV location and we define the social cost ratio by comparing to the social optimum:
$\mathtt{Ratio.1}=\frac{SC(med(\textbf{x},\textbf{y},z_0),(\textbf{x},\textbf{y}))}{OPT_1(\textbf{x},\textbf{y})}.$
In Mechanism \ref{m2}, we choose weighted median location as the UAV location and we define the social cost ratio:
$\mathtt{Ratio.2}=\frac{SC(wmed(\textbf{x},\textbf{y},z_0),(\textbf{x},\textbf{y}))}{OPT_1(\textbf{x},\textbf{y})}.$
Note that $\mathtt{Ratio.1}$ and $\mathtt{Ratio.2}$ are random variables, depending on distributions of $\textbf x$, $\textbf{y}$ and $\textbf{z}$, while approximation ratio $\gamma$ characterizes the maximum of each ratio in the worst-case.



\begin{theorem}\label{t3}
	For $\alpha=2,$ as the number $n$ of users goes to infinity, both $\mathtt{Ratio.1}$ and $\mathtt{Ratio.2}$ converge in probability towards $1$, given all $x_i$'s, $y_i$'s, $z_i$'s and $w_i$'s are independent and identically distributed, respectively, and all $x_i$'s, $y_i$'s and $z_i$'s follow any continuous symmetric distributions (including normal, uniform and logistic distributions), respectively.
\end{theorem}

\begin{figure}[h]
\centering
\subfigure[Mean social cost of Mechanism 1 and mean social cost of Mechanism 2 versus mean optimal social cost.]{
\label{fig1.1a} \includegraphics[height=7.0cm, width=8.5cm]{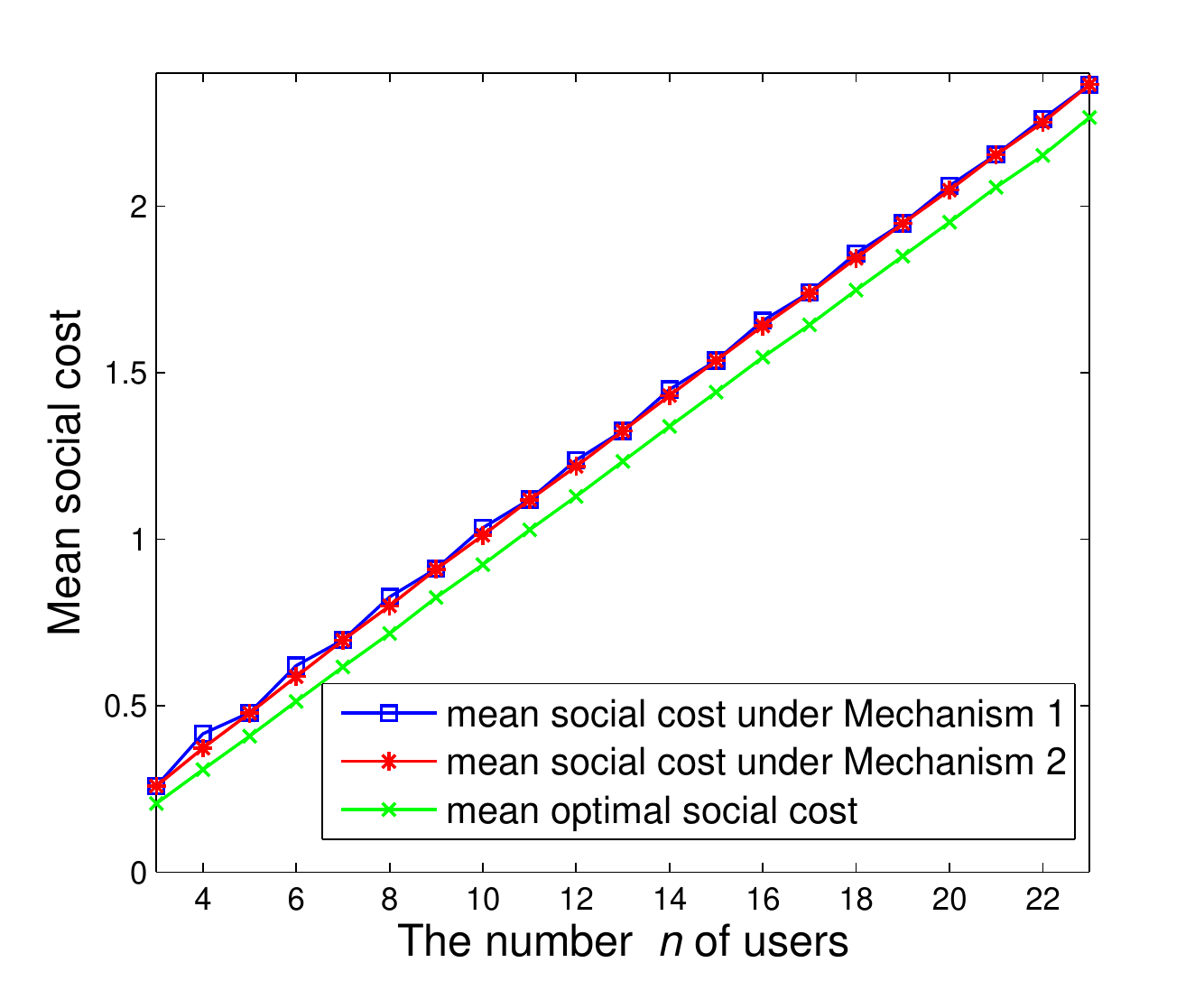}}
\hspace{-0.00in}
\subfigure[Mean $\mathtt{Ratio.1}$ of Mechanism 1 and mean $\mathtt{Ratio.2}$ of Mechanism 2 versus the number $n$ of users.]{
\label{fig1.1b}
\includegraphics[height=7.0cm, width=8.5cm]{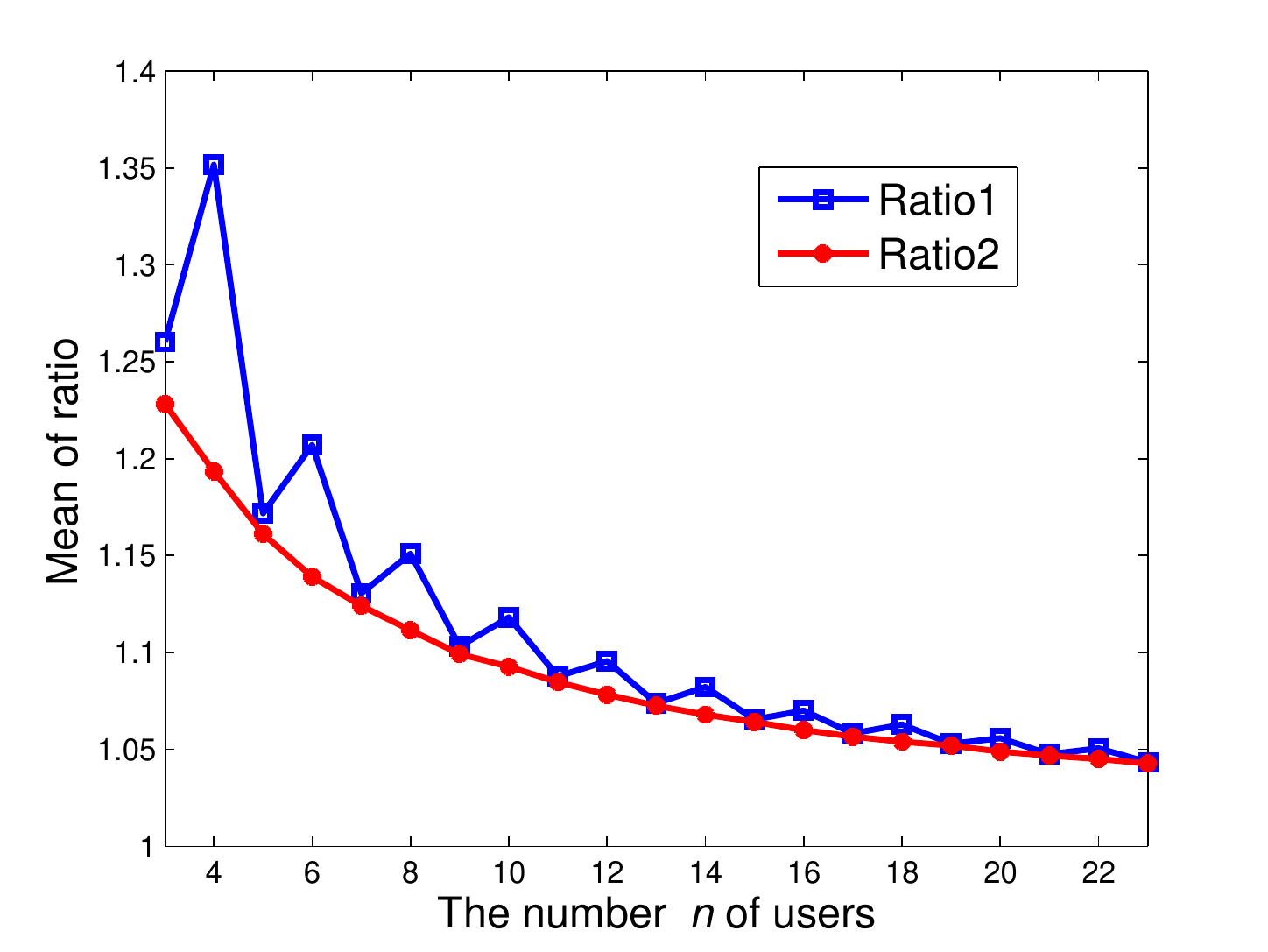}}

\caption{ Simulations for comparing the social costs under
Mechanisms 1 and 2 with the optimal social cost.} \label{fig2}
\end{figure}

The proof of Theorem \ref{t3} is given in Appendix  \ref{app_3}.
Theorem \ref{t3} proves that the two mechanisms perform
optimally when $n$ is sufficiently large.
We next provide more  simulations for evaluating Mechanisms \ref{m1} and \ref{m2}. For simplicity, we assume $\alpha=2,$ $z_0=0.2$ and $I^2=[0,1]^2$, where each user's location follows the continuous uniform distribution in $I^2$ and each $w_i$ follows the continuous uniform distribution in $[0,1]$.

Fig.\ref{fig1.1a} shows that
the two mean social costs under Mechanisms \ref{m1} and  \ref{m2} as well as the mean optimal social cost increase linearly as the number of users increases.
No matter which number of users we are looking at, the performance gaps between  the mean optimal social cost and the mean social costs achieved by Mechanism \ref{m1} and  \ref{m2} are quite small, which means our Mechanisms \ref{m1} and  \ref{m2}  approximate well the optimal solution in average sense.

Fig.\ref{fig1.1b} shows that the mean two random  ratios 
decrease to $1$ as the number of users increases.
Consistent with our prior worst-case conclusion, Fig.\ref{fig1.1b} also shows that in the average-case
Mechanism \ref{m2} outperforms Mechanism \ref{m1}, as the mean of $\mathtt{Ratio.2}$ is smaller than the mean of $\mathtt{Ratio.1}$ given any number of users.
Yet such advantage is no longer obvious once $n>10$.
Interestingly, from Fig.\ref{fig1.1b} we can observe that $\mathtt{Ratio.1}$ at odd number $n$ ($n=2m-1$ with natural number $m$) of user size is smaller than $\mathtt{Ratio.1}$ at neighboring even number $n$ ($n=2m$). This is because if $n$ is odd, $med(\textbf{x},\textbf{y}, z_0)$ can be relatively closer to $(\bar{x},\bar{y},z_0),$ as compared to  the case that $n$ is even.

\section{Obnoxious UAV Placement Game for Type 2 Users }\label{sectionoflg}
In this section, we study strategyproof mechanism design for the obnoxious UAV placement game, where all $n$ users of type 2 in Fig.\ref{fig1} view the UAV obnoxious due to its introduced interference and want to be far away from the UAV.
The optimization problem of this game  is formulated as
\begin{align}
\begin{cases}
	&\max_{f}  \sum_{i=1}^n w_i ((x_i-x)^2+(y_i-y)^2+z_0^2)^{\alpha/2},\nonumber\\
& \mbox{s.t.
for any misreported location $(x_i' ,y_i')\!\in\! I^2\!$ and user $i\!\in\! N$,}\nonumber\\
&\ \ \ \ u_i(f((x_i,y_i),(\textbf x_{-i},\textbf y_{-i})),(x_i, y_i))\nonumber\\
&\geq u_i(f((x_i'  ,y_i'),(\textbf x_{-i},\textbf y_{-i})),(x_i, y_i)).\nonumber\\
&\mbox{Variable: function } f(\textbf x, \textbf y)=(x,y): I^{2n} \to I^2.\nonumber\\
&\mbox{Parameters:} (x_i, y_i)\in I^2,w_i>0, \mbox{for } i=1,\dots,n, z_0\geq 0,  \nonumber\\
&\mbox{and }\alpha\in[2,6].\nonumber
\end{cases}
\end{align}

\subsection{Design and analysis of strategyproof mechanism}
We first analyze the optimal UAV location under full information as the benchmark.
We obtain that
\begin{align}
&\frac{\partial SU(x,y)}{\partial x}
\!=\!\sum_{i=1}^n\alpha w_i(x\!-\!x_i)((x\!-\!x_i)^2\!+(y\!-\!y_i)^2\!+z_0^2)^{\frac{\alpha-2}{2}},\nonumber\\
&\frac{\partial^2 SU(x,y)}{\partial x^2}=\sum_{i=1}^n\alpha w_i(((x-x_i)^2+(y-y_i)^2+z_0^2)^{\alpha/2-1}\nonumber\\
&+(\alpha-2) ((x-x_i)^2+(y-y_i)^2+z_0^2)^{\alpha/2-2}(x-x_i)^2)\geq 0,\nonumber
\end{align}
and similarly, $\frac{\partial^2 SU(x,y)}{\partial y^2}\geq 0.$
Thus $SU((x,y),(\textbf x, \textbf y))$ is convex in $(x,y)\in I^2.$
The optimal (maximum) solution should lie in boundary of $I^2$, which are the corners of rectangle $I^2.$
We obtain the optimal UAV location is $x=0$ or $2A$ and $y=0$ or $2B.$

This optimal solution is not strategyproof by considering  an illustrative example in 1D: given $\alpha=2,$ there are user $1$ at $x_1=0.2$ and user $2$ at $x_2=0.6$ in domain 1D $I=[0,1]$, where user $2$ can misreport his location to $x_2'=1$ to keep the UAV away from him at $x=0.$
Next, we design strategyproof mechanisms.
\begin{mechanism}\label{m4}
	The UAV strategically decides its location $f=(x,y,z_0)$, where
	\begin{align}
& x= \left\{ \begin{array}{ll}
	0, & \mbox{if $\sum_{i:x_i\in [0,A)}w_i\leq\sum_{i:x_i\in [A,2A]}w_i$};\\
	2A, & \mbox{otherwise},\end{array} \right.\nonumber\\
& y= \left\{ \begin{array}{ll}
	0, & \mbox{if $\sum_{i:y_i\in [0,B)}w_i\leq\sum_{i:y_i\in [B,2B]}w_i$};\\
	2B, & \mbox{otherwise}.\end{array} \right.
\end{align}
\end{mechanism}

In Mechanism \ref{m4}, The UAV compares the total user weights in regimes $[0,A)$ and $[A,2A]$ of the $x$-domain, and places the obnoxious UAV at the corner with the smaller total weight. Similarly, The UAV places its location in $y$-domain for the weighted majority's benefit.

\begin{theorem}\label{t4}
	Mechanism \ref{m4} is a strategyproof mechanism with approximation ratio $\gamma=5\times 2^{(\alpha-2)/2}$ in the obnoxious UAV placement game.
\end{theorem}
The proof of Theorem \ref{t4} is given in Appendix \ref{app_4}.

\begin{mechanism}\label{m1.1}
	The UAV strategically decides its location $f=(x,y,z_0)$, where
	\begin{align}
& x= \left\{ \begin{array}{ll}
	0, & \mbox{if $\sum_{i:x_i\in [0,A)}1\leq\sum_{i:x_i\in [A,2A]}1$};\\
	2A, & \mbox{otherwise},\end{array} \right.\nonumber
\end{align}
\begin{align}
& y= \left\{ \begin{array}{ll}
	0, & \mbox{if $\sum_{i:y_i\in [0,B)}1\leq\sum_{i:y_i\in [B,2B]}1$};\\
	2B, & \mbox{otherwise},\end{array} \right.\nonumber
\end{align}
\end{mechanism}

	Mechanism \ref{m1.1} is a strategyproof mechanism with approximation ratio $\gamma=5\frac{w_{max}}{w_{min}}2^{(\alpha-2)/2}$ in the obnoxious UAV placement game.
Mechanism \ref{m1.1} treats each user equally and does not consider
users' weights.
Since the UAV does not need to
gather the information of weights from users, it is strategyproof
even if we allow users to misreport their weights.
\subsection{Empirical analysis of Mechanism \ref{m4}}
In this subsection, we present empirical analysis to evaluate the average performances of Mechanism \ref{m4} given $\alpha=2$.
We have $SU(f,(\textbf x,\textbf y))=\sum_{i=1}^n w_i ((x_i-x)^2+(y_i-y)^2+z_0^2).$
We split social utility as $SU=SU_x+SU_y+\sum_{i=1}^n w_i z_0^2.$ 
By using weighted mean $\bar{x}$ in (\ref{a17}), we rewrite  the social utility (\ref{a15}) in $x$-domain as
\begin{align}
&SU_x(f,\textbf x)=\sum_{i=1}^n w_i ((x_i-\bar{x})+(\bar{x}-x))^2\nonumber\\
&=\sum_{i=1}^n w_i((x_i-\bar{x})^2+(\bar{x}-x)^2)
+2(\bar x-x)\sum_{i=1}^n w_i(x_i-\bar x),\nonumber
\end{align}
where the last summation term 
is zero due to (\ref{a17}).
Thus, we can rewrite 
$SU_x(f,\textbf x)=(x-\bar{x})^2\sum_{i=1}^n w_i+\sum_{i=1}^n w_i(x_i-\bar{x})^2.$
Similarly, we can obtain $SU_y$ and finally $SU$ as
\begin{align}
&SU(f,(\textbf x,\textbf y))
=((x-\bar{x})^2+(y-\bar{y})^2)\sum_{i=1}^n w_i\nonumber\\
&+\sum_{i=1}^n w_i ({(x_i-\bar{x})}^2+{(y_i-\bar{y})}^2+z_0^2). \nonumber
\end{align}
We can see that $SU(f,(\textbf x,\textbf y))$ is linear with the square of the distance between 
$(x,y)$ and
$(\bar{x},\bar{y})$.
Obviously, we obtain the optimal UAV location as
\begin{equation} \label{a9}
x_{opt}=\left\{ \begin{array}{ll}
0, & \mbox{if $\bar{x}\geq A$};\\
2A, & \mbox{if $\bar{x}<A$},\end{array} \right.
y_{opt}=\left\{ \begin{array}{ll}
0, & \mbox{if $\bar{y}\geq B$};\\
2B, & \mbox{if $\bar{y}<B$}.\end{array} \right.
\end{equation}

\begin{theorem}\label{t7}
	Given $\alpha=2$, as the number $n$ of users  goes to infinity, the probability that UAV location $(x,y,z_0)$ under Mechanism \ref{m4} equals the social optimal location goes to $1$, given all $x_i$'s,  $y_i$'s and $w_i$'s are independent and identically distributed, respectively, and all $x_i$'s, and $y_i$'s follow any continuous asymmetric distributions (including Beta distribution and skew normal distribution), respectively.
\end{theorem}

\begin{figure}
\centering
\subfigure[Distributions of users' locations $x_i$'s in $x$-domain]{
\label{fig3a}
\includegraphics[height=7.0cm, width=8.5cm]{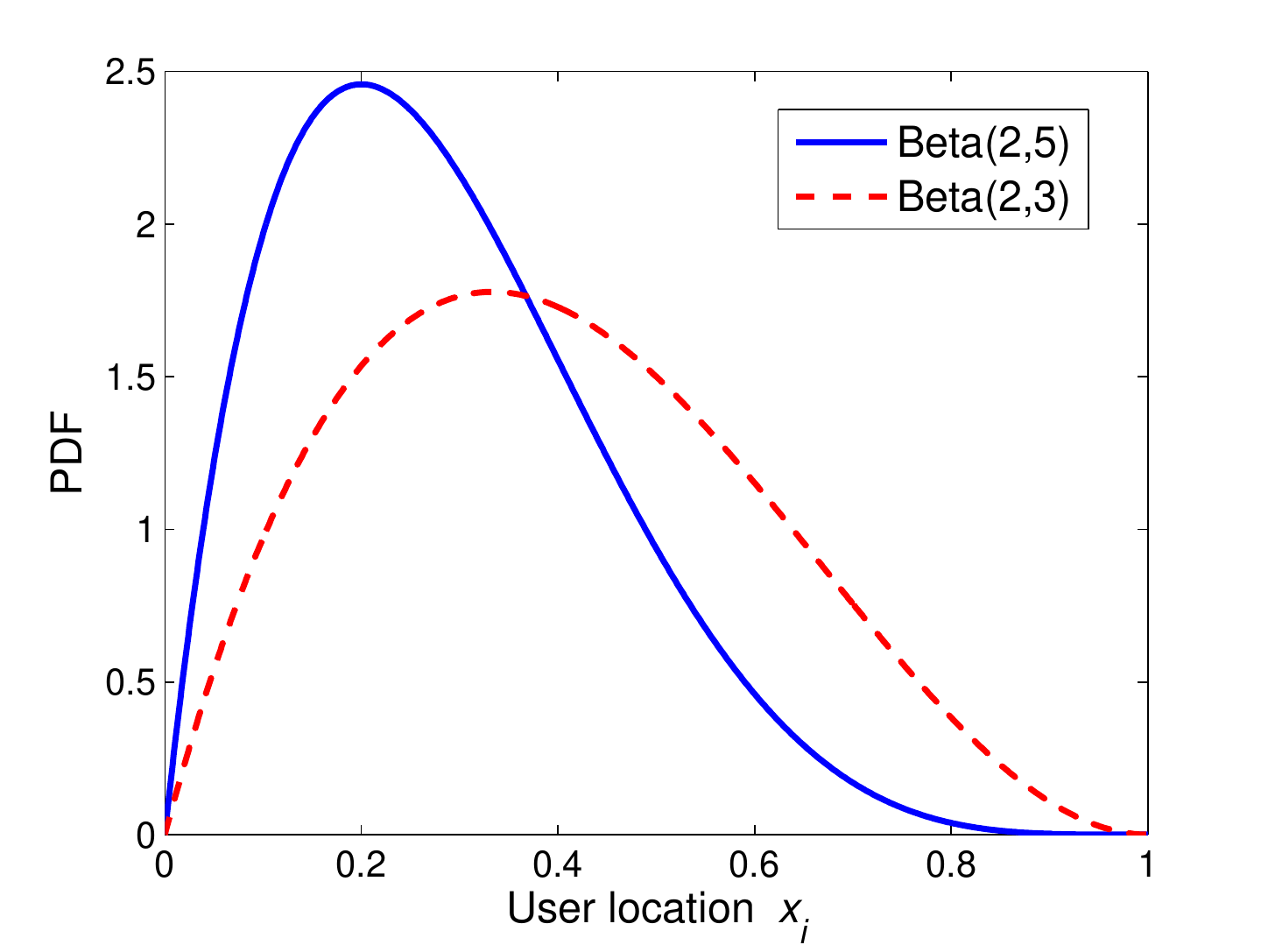}}

\hspace{-0.00in}

\subfigure[Probability that Mechanism \ref{m4} is equal to the optimal location]{
\label{fig3b} \includegraphics[height=7.0cm, width=8.5cm]{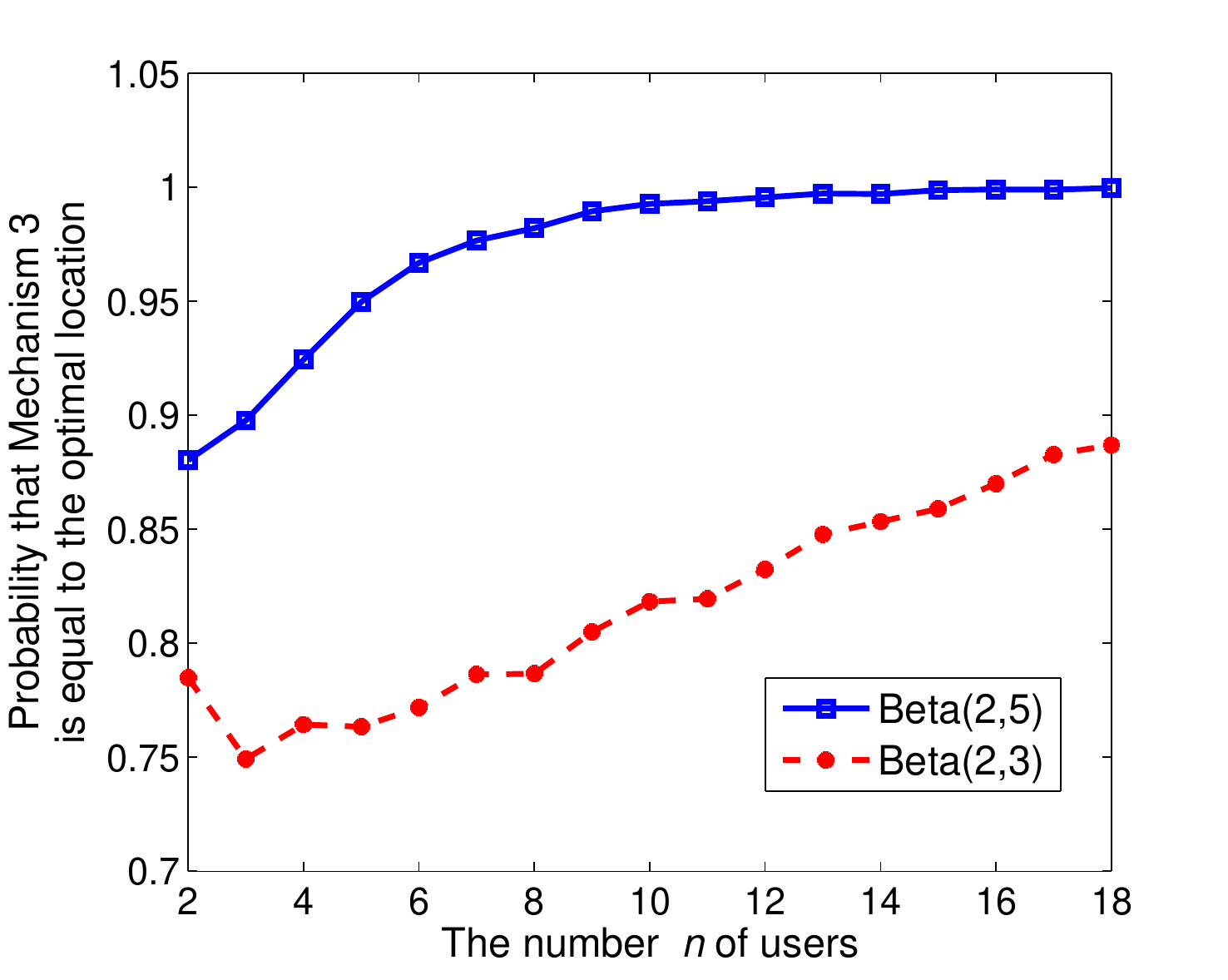}}

\hspace{-0.00in}

\subfigure[Mean social utility of Mechanism \ref{m4} versus mean optimal social utility under Beta(2,5) and Beta(2,3).]{
\label{fig3c} \includegraphics[height=7.0cm, width=8.5cm]{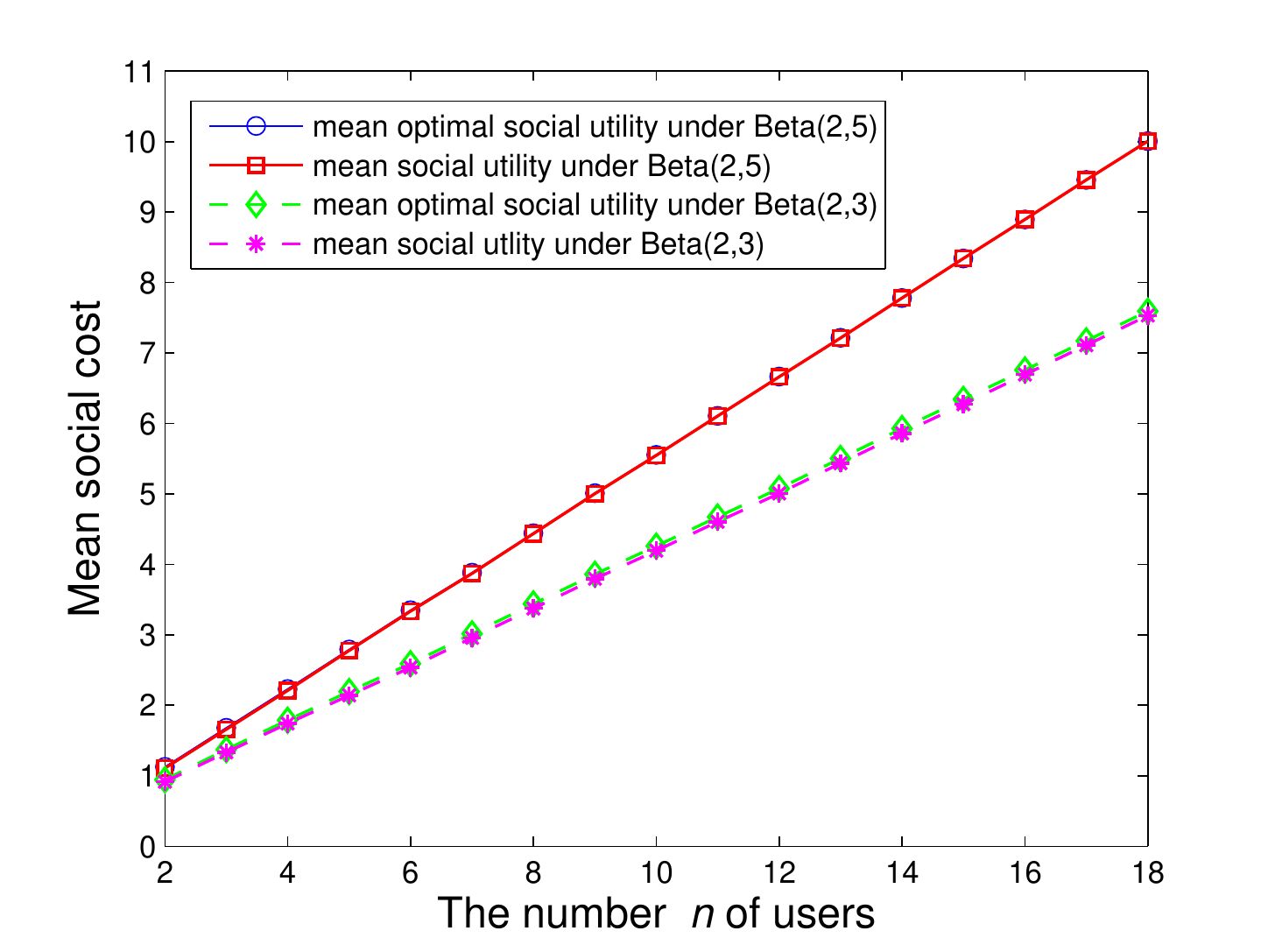}}

\caption{Skewness of two asymmetric random distributions of users' locations, and Mechanism \ref{m4}'s convergence rate to the optimal placement under the two distributions and the number $n$ of users.} \label{fig3}
\end{figure}

  \begin{proof}
	We only consider $x$-domain location analysis due to symmetry.
	Compared with $x_{opt}$ in (\ref{a9}), Mechanism \ref{m4} in $x$-domain can be rewritten as:
	\begin{align}
	x= \left\{ \begin{array}{ll}
	0, & \mbox{if } x_{wmed}\geq A;\\
	2A, & \mbox{otherwise}.\end{array} \right.\nonumber
	\end{align}
	In the proof of Theorem \ref{t3} in Appendix \ref{app_3}.
we have proved that $\bar{x}$ converges in probability towards the expectation of $x_1$
	(i.e., $\bar{x}\stackrel{P}{\rightarrow}\E[x_1]$, as $n\rightarrow \infty$) and $x_{wmed}$ converges in probability towards the median of $x_1$
	(i.e., $x_{wmed}\stackrel{P}{\rightarrow}H^{-1}(\frac{1}{2})$, as $n\rightarrow \infty$).
	Since all $x_i$'s follow continuous asymmetric distribution, respectively, mean of $x_i$'s and median of $x_i$'s lie in either $[0,A)$ or $[A,2A]$ simultaneously.
	Without loss of generality, we consider that they all lie in  interval $[0,A)$.
	 Considering the optimal UAV location, we have $\E[x_1]\in$ $ [0,A)$ and thus $\bar{x} \stackrel{P}{\in}$$ [0,A)$, as $n\rightarrow \infty$. According to (\ref{a9}), the optimal UAV location  satisfies
	$x_{opt}\stackrel{P}{\rightarrow}2A$, as $n\rightarrow \infty$.
	Considering UAV location in Mechanism \ref{m4}, we have $H^{-1}(\frac{1}{2})$ $\in[0,A)$ and thus $x_{wmed}\stackrel{P}{\in}$$ [0,A)$, as $n\rightarrow \infty$.
	According to Mechanism \ref{m4}, the optimal UAV location  is that
	$x\stackrel{P}{\rightarrow}2A$, as $n\rightarrow \infty$.
	Therefore, 	as $n\rightarrow \infty$, $\Pr[x=x_{opt}]{\rightarrow}1$ in Mechanism \ref{m4}.
\end{proof}

Comparing Theorems \ref{t3} and \ref{t7}, we can see that  Theorem \ref{t3} needs the symmetric location distribution condition for type 1 users, which makes sure the UAV location in Mechanisms \ref{m1} and \ref{m2} can approach the optimal point. While  Theorem \ref{t7} needs asymmetric location distribution condition for type 2 users, which makes sure the UAV location in both Mechanism \ref{m4} and the optimal point can diverge in the same direction towards the same corner.

We  provide empirical simulations in Fig.\ref{fig3} for Mechanism \ref{m4} when $n$ is finite. For simplicity, we assume $I^3=[0,1]^3$, each user's location follows asymmetric Beta distribution in $I^3$ and every $w_i$ follows the continuous uniform distribution in $[0,1]$.

We have two user groups' location distributions for simulation comparisons.
We can see from Fig.\ref{fig3a} that random distribution Beta(2,5) has larger skewness and is more asymmetric
than Beta(2,3), and thus provides faster convergence rate for Mechanism \ref{m4} towards the social optimum, as observed from  Fig.\ref{fig3b}.  This is consistent with Theorem \ref{t7}. Intuitively, a larger skewness of users' distribution implies a higher probability for UAV location of Mechanism \ref{m4} to be equal to the optimal UAV location.

Fig.\ref{fig3c} shows that
the mean social utility of Mechanism \ref{m4} versus the mean optimal social utility under distributions Beta(2,5) and Beta(2,3).
The mean social utility achieved by our Mechanism \ref{m4} and the mean optimal social utility under both Beta(2,5) and Beta(2,3) increase linearly as the number of users increases.
No matter which number of users we are looking at, the performance gap between the mean social utility achieved by Mechanism \ref{m4} and the mean optimal social utility is quite small, telling that our Mechanism \ref{m4} approximates the optimal solution well in average sense.

\section{Dual-Preference UAV Placement Game for Both Types of Users }\label{sectiondpflg}

In this section, we design a strategyproof mechanism in the dual-preference UAV placement game where both types of users co-exist. Without loss of generality, we assume all users' weights as $1$, i.e., $w_i=1$ for any $i\in N$, and our results (though more involved)  can also be extended to the weighted case in any distributions.

As shown in Fig.\ref{fig1}, each user has his own preference type and we define user $i$'s type as $\theta_i$ which is either $1$ or $2$.
 A user $i$ with $\theta_i=1$ (facility user)  prefers  to be close to the UAV and a user $i$  with $\theta_i=2$ (adverse user) prefers  to be far away from the UAV.
We denote $\Theta=\{\theta_1,\dots,\theta_n\}$ as the profile of all $n$ users' preferences.
Now the UAV needs to gather information of users' preference types  besides users' locations to determine $(x,y,z_0).$
Note that a user may also cheat on his report of preference type, and this adds difficulty to the UAV's strategyproof mechanism design. 	
Given the location of the UAV $(x,y,z_0)$, we define a user $i$'s utility as
\begin{align}
&u_i((x,y,z_0), (x_i,y_i,\theta_i))\nonumber\\
=&\begin{cases}
((2A)^2-(x_i-x)^2+(2B)^2-(y_i-y)^2+z_0^2)^{\alpha/2},
\nonumber\\
\text{\ \ \ \ \ \ \ \ \ \ \ \ \ \ \ \ \ \ \ \ \ \ \ \ \ \ \ \ \ \ \ \ \ \ \ \ \ \ \ \ \ \ \ \ \ \ \ \ \ \ \ \ \ \ \ \ \ \ \ \ \ \ \ \ \  if }\theta_i = 1;\nonumber\\
 ((x_i-x)^2+(y_i-y)^2+z_0^2)^{\alpha/2}, \text{\ \ \ \ \ \ \ \ \ \ \ \ \ \  if } \theta_i = 2.\nonumber\\
 \end{cases}
\end{align}

Note that we want to minimize the service cost for a type 1 user (facility user) as in Section \ref{sectionflg}, which is equivalent to maximizing the user's utility  $-(x_i -x)^2-(y_i -y)^2$.
To make our definition of approximation ratio meaningful, we require
nonnegative utilities and add $(2A)^2+(2B)^2$ to the utility. This technique is widely used (e.g., \cite{zou2015facility}) and does not change our main results.

\begin{definition}\label{d3}
	A mechanism is strategyproof in the dual-preference UAV placement game if no user can benefit from misreporting his location and preference type. Formally, given location profile $\Omega=(< x_i,\textbf x_{-i}>, <y_i,\textbf y_{-i}>)\in I^{2n}$, preference profile $\Theta$, and any misreported location $(x_i'  ,y_i')\in I^2$ and preference type $\theta_i'$ for user $i\in N$, it holds that
\begin{align}
	&u_i( f((x_i,y_i,\theta_i),({\textbf x_{-i}},\textbf y_{-i},\Theta_{-i})),(x_i,y_i,\theta_i))\nonumber\\
	\geq & u_i( f((x_i',y_i',\theta_i' ),(\textbf x_{-i},\textbf y_{-i},\Theta_{-i})),(x_i,y_i,\theta_i)). \nonumber
\end{align}
\end{definition}

Given a location profile $\Omega$, let $OPT_3(\textbf x,\textbf y, \Theta)$ be the optimal social utility in this game.
A strategyproof mechanism $f$ has an approximation ratio $\gamma\geq 1$, if for any location profile $(\textbf x,\textbf y, \Theta)$ and $\Theta,$ $OPT_3(\textbf x,\textbf y, \Theta)\leq \gamma SU(f,(\textbf x,\textbf y, \Theta))$.

The objective of this game is to $\max_{x,y} \sum_{i=1}^n $ $u_i((x,y,z_0),$ $(x_i,y_i,\theta_i)).$
We can see that the social utility function is quadratic and it is not difficult to derive the optimal UAV location by checking the first-order condition.
However, outputting the optimal location is not a strategyproof mechanism and we needs to design a stragetyproof mechanism.

\subsection{Design and analysis of strategyproof mechanism}

\begin{mechanism}\label{m7}
	Consider $x$-domain and define two user sets for each preference type:
 \begin{align}
 &R_{1}=\{i :\theta_i = 1, x_i\geq A\}, R_{2}= \{i:\theta_i = 2, x_i> A\},\nonumber\\
 &L_{1}=\{i :\theta_i = 1, x_i< A\}, L_{2} = \{i:\theta_i = 2, x_i\leq A\}.\nonumber
 \end{align}
The $x$-location of the UAV is $x=2A$ if $|R_{1}|+|L_{2}|\geq|R_{2}|+|L_{1}|$ and
	$x=0$ otherwise.
Consider $y$-domain and define two user sets for each preference type:
 \begin{align}
 &\bar R_{1}=\{i :\theta_i = 1, y_i\geq B\}, \bar R_{2}= \{i:\theta_i = 2, y_i> B\},\nonumber\\
 &\bar L_{1}=\{i :\theta_i = 1, y_i< B\}, \bar L_{2} = \{i:\theta_i = 2, y_i\leq B\}.\nonumber
 \end{align}
The $y$-location of the UAV is $y=2B$ if $|\bar R_1|+|\bar L_2|\geq|\bar R_{2}|+|\bar L_{1}|$ and
	$y=0$ otherwise.
\end{mechanism}

In Mechanism \ref{m7}, users in sets $R_1$ and $L_2$ prefer the UAV to locate at $x=2A$, while the other users prefer the UAV to locate at $x=0$. We follow the majority rule to design Mechanism \ref{m7}.
Actually, Mechanism \ref{m7} for the dual-preference UAV placement game is inspired by Mechanism \ref{m4} for the obnoxious UAV placement game, yet this new mechanism further considers that users may also cheat on their preference types besides locations.

\begin{theorem}\label{t8}
	Mechanism \ref{m7} is a stragegyproof mechanism with approximation ratio $\gamma=2^{3\alpha/2}$.
\end{theorem}
The proof of Theorem \ref{t8} is given in Appendix \ref{app_5}.


\subsection{Empirical analysis of Mechanism \ref{m7}}

\begin{figure}[t]
	\centerline{\includegraphics[height=7.0cm, width=8.5cm]{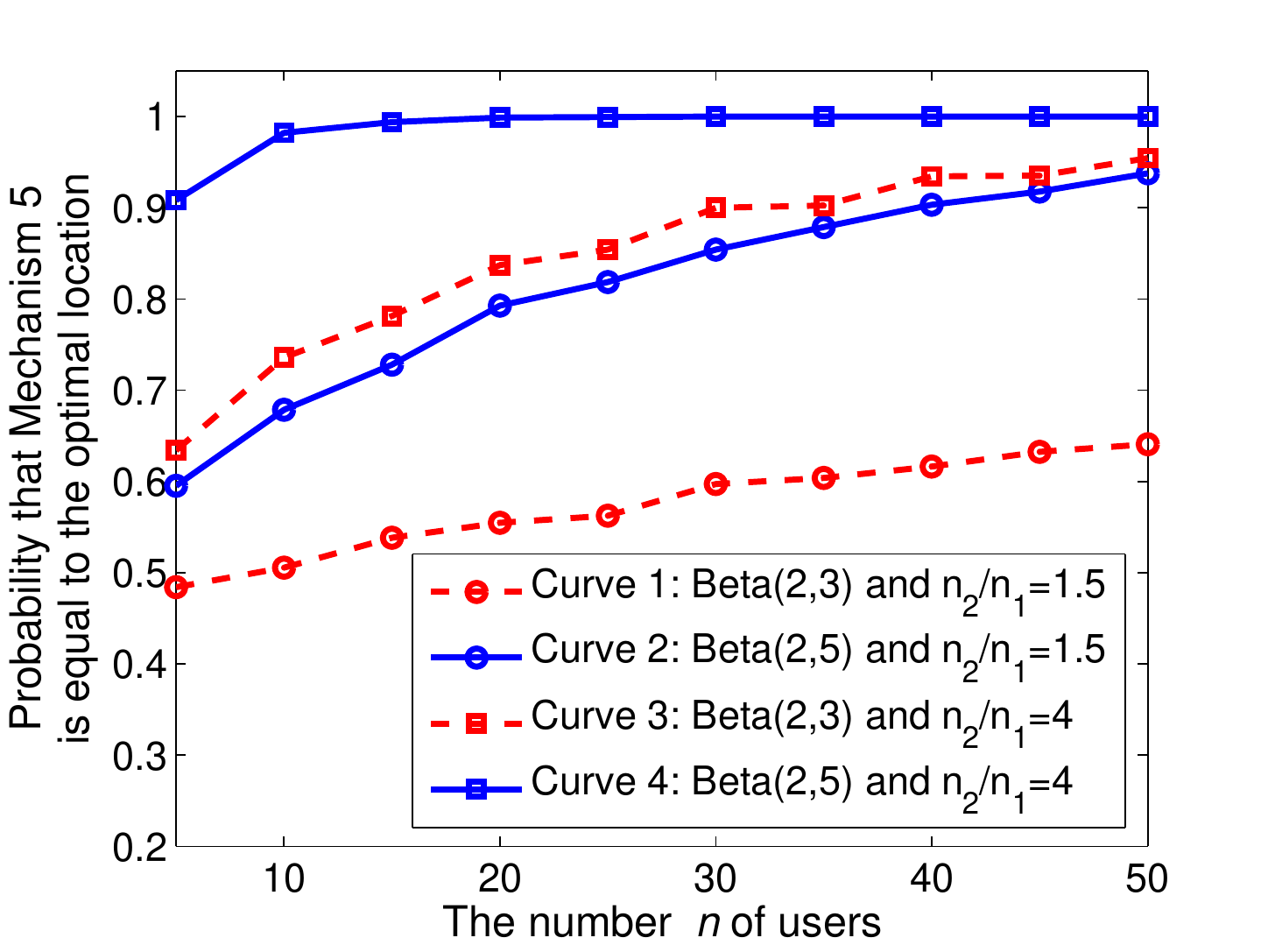}}
	\caption{Mechanism \ref{m7}'s convergence to the optimal placement under two different asymmetric distribution of users' locations and  two different $n_2/n_1$ value of Mechanism \ref{m7} versus convergence rate to the optimal placement.}
\label{fig4}
\end{figure}

In this subsection, we present empirical analysis to evaluate the average performances of Mechanism \ref{m7} for $\alpha=2$. Define $n_1=|\{i:\theta_i=1\}|$, $n_2=|\{i:\theta_i=2\}|$ with $n_1+n_2=n$.

We  provide empirical simulations in Fig.\ref{fig4} for Mechanism \ref{m7} when $n$ is finite.
For simplicity, we assume $I^2=[0,1]^2$, each user's location follows asymmetric Beta distribution in $I^2$ and $n_2/n_1$ is greater than $1$.
We have four curves in the simulations for comparisons. In Curves 1 and  2, $n_2/n_1=1.5,$ and  in Curves 3 and 4, $n_2/n_1=4$ with more type 2 users given the same number of total users. In Curves 1 and 3, each user's location follows Beta(2,3), and in Curves 2 and 4, each user's location follows more asymmetric distribution Beta(2,5) in Fig.\ref{fig3a}.

We can see in Fig.\ref{fig4} that, distribution Beta(2,5) provides faster convergence rate for Mechanism \ref{m7} to approach 1 (social optimum) by comparing  Curves 1 and 2 (or Curves 3 and 4).
Intuitively, a larger skewness of users' distribution tells a higher probability of UAV location under Mechanism \ref{m7} equal to the optimal UAV location.
By comparing  Curves 1 and 3 (or Curves 2 and 4), we can see that the higher value of $n_2/n_1$ with more type 2 users provides faster convergence rate for Mechanism \ref{m7} to approach the social optimum.
As shown in Fig.\ref{fig3c}, we have shown that Mechanism \ref{m4} already
approximates the optimal solution so we omit showing the similar figure with benchmarks here.

\section{Extensions to Multi-UAV Placement}\label{section_multiple}
 In previous sections, we limit our design to a single UAV's deployment.  In this section, we extend the mechanism design of the three placement games by considering multiple cooperative UAVs.
For simplification,
we only consider $\alpha= 2.$
We first study the dual-preference  placement game for two UAVs, and then study the rest two one-preference games for an arbitrary number $k$ of UAVs.

\subsection{The placement game with two dual-preference UAVs}
In this subsection, we design a strategyproof mechanism
for the two-UAV location game with dual-preference for $\alpha= 2.$
Though we only have two UAVs here, we allow the general type setting, where each user has his own preference towards one out of the two different UAVs and a user's preferences over the two UAVs may not be the same. Technically, it is difficult to tackle this more general case with more than one UAV as each user may have diverse
and hidden preferences over different UAVs.
We define preference of  user $i$ to UAV $j$ as $\theta_i^j$ which is either $1$ or $2$ and user $i$ may not truthfully report his preferences. Any user
$i$ with $\theta_i^j=1$ prefers to be close to UAV $j$ and
any user $i$ with $\theta_i^j=2$  prefers to be far away from
 UAV $j$, where $j=1,2$. We denote $\Theta=\{(\theta_1^1,\theta_1^2),(\theta_2^1,\theta_2^2),\dots,(\theta_n^1,\theta_n^2)\}$ as the profile of all $n$
users' preferences. The UAVs need to gather information
of both users' locations $(\textbf x,\textbf y)$ and  preferences $\Theta$ to jointly determine their locations $(X_1,Y_1,Z_1)$ and $(X_2,Y_2,Z_2).$
Given the locations of the two UAVs $((X_1,Y_1,Z_1),(X_2,Y_2,Z_2))=f(\textbf x,\textbf y, \Theta)$, similar to Section \ref{sectiondpflg}, we define user $i$'s utility towards UAV $j$  as
\begin{align}\label{a113}
 u_i^j= \left\{ \begin{array}{ll}
	(2A)^2-(X_j-x_i)^2+(2B)^2\\
-(Y_j-y_i)^2, &\mbox{if } \theta_i^j=1;\\
	(X_j-x_i)^2+(Y_j-y_i)^2+Z_j^2,  &\mbox{if } \theta_i^j=2.\\
\end{array} \right.
\end{align}

The user $i$'s utility is $u_i=u_i^1+u_i^2.$
The social utility of a mechanism $f$ is defined as:
$SU(f((\textbf x,\textbf y), \Theta),((\textbf x,\textbf y), \Theta))=\sum_{i=1}^n u_i(f((\textbf x,\textbf y),\Theta),x_i,\theta_i).$
We still use $OPT_3((\textbf x,\textbf y),\Theta)$ to denote the optimal solution of maximizing $SU$.

We say a mechanism $f$ has an approximation ratio $\gamma$, if for any profile $(\textbf x,\textbf y)\in I^{3n}$ and $\Theta\in \{1,2\}^{2n}$, $ OPT_3((\textbf x,\textbf y),\Theta)\leq \gamma SU(f,(\textbf x,\textbf y),\Theta)$.
To maximize $SU$ and  obtain the optimal solution $OPT$ is difficult,  since $SU$ is neither convex nor concave now. This is different from Section \ref{sectiondpflg} for a single UAV case and adds difficulty to mechanism design. Still, we can show that the optimal solution is not strategyproof since this dual-preference game's special case is the obnoxious UAV placement game.

By considering $n$ users' different $x$-domain locations and mixed preferences towards the two UAVs, we define eight sets of users as shown in Table \ref{setofq}  with $Q_1\cup Q_2\dots\cup Q_8=N$:
\begin{table}[!htbp]
\centering
\caption{Sets   of users $Q_1- Q_8$.}\label{setofq}
\begin{tabular}{|c|c|c|c|c|c|}
\hline

\multicolumn{2}{|c|}{ \multirow{2}*{Set of users} }& \multicolumn{4}{c|}{$\{\theta_i^1,\theta_i^2\}=$} \\
\cline{3-6}
\multicolumn{2}{|c|}{}&\{1,1\}&\{1,2\}&\{2,1\}&\{2,2\}\\
\hline
\multirow{2}*{$x_i\in$}
&$[0,A]$&$Q_1$&$Q_2$&$Q_3$&$Q_4$\\
\cline{2-6}
&$(A,2A]$&$Q_5$&$Q_6$&$Q_7$&$Q_8$\\
\hline
\end{tabular}
\end{table}

We rewrite social utility  as $SU=SU_x+SU_y+SU_z$, where $SU_x(f(\textbf x, \Theta),(\textbf x, \Theta))$ is as follows based on (\ref{a113}) and Table \ref{setofq}:
\begin{align}\label{a131}
 &SU_x(f(\textbf x, \Theta),(\textbf x, \Theta))=\sum_{i=1}^n (u_{i,x}^1+u_{i,x}^2)\nonumber\\
 =&\sum_{i\in Q_1\cup Q_5}(8A^2-(X_1-x_i)^2-(X_2-x_i)^2)\nonumber\\
 &+\sum_{i\in Q_2\cup Q_6}(4A^2-(X_1-x_i)^2+(X_2-x_i)^2)\nonumber\\
&+\sum_{i\in Q_3\cup Q_7}((X_1-x_i)^2+4A^2-(X_2-x_i)^2)\nonumber\\
&+\sum_{i\in Q_4\cup Q_8}((X_1-x_i)^2+(X_2-x_i)^2).
\end{align}
Next, we present our mechanism for the two UAVs.
\begin{mechanism}\label{m110}
If  $|Q_2|+|Q_7|\geq |Q_3|+|Q_6|$, the two UAVs locate their $x$-domain locations to $(X_1, X_2)$ $=(0,2A);$
Otherwise, $(X_1, X_2)=(2A,0).$ $(Y_1, Y_2)$ in the $y$-domain follows the same placement strategy, respectively.
\end{mechanism}

\begin{theorem}\label{t110}
Mechanism \ref{m110} is strategyproof with approximation ratio $\gamma=4$.
\end{theorem}
The proof of Theorem \ref{t110} is  given in Appendix  \ref{app_6}.
If there are $k>2$ UAVs in the dual-preference placement game, the mechanism design and  analysis will be more involved by including $2^{k+1}$ user sets and are left for future study. Now, we only have $8$ user sets $Q_1 - Q_8$ for $k=2.$ Next, we focus on mechanism design for one-preference $k$-UAV placement games.

\subsection{The obnoxious and favorable $k$-UAV placement game}

 We now study the  obnoxious multi-UAV placement game for $\alpha= 2.$ There are $k\geq 2$ obnoxious UAVs, where each user of type 2  wants to maximize his utility to keep  away from the UAVs.
The location of the $j$-th UAV is $(X_j, Y_j, Z_j)$ with $j=1,2,\dots, k.$
Since each user $i$ wants to reduce the total  interference from all the UAVs,   the utility of user $i$ should be the total utility of user $i$ given $k$ UAVs' locations, i.e.,
\begin{align}
	u_i(f(\textbf x,\textbf y),(x_i,y_i))
= \sum_{j=1}^k w_i ((x_i-X_j)^2+(y_i-Y_j)^2+Z_j^2).\nonumber
\end{align}
Similar to one obnoxious UAV placement game in Section \ref{sectionoflg}, the social utility is the total utility of $n$ users and $k$ UAVs cooperate to maximize the social utility by jointly choosing their locations.
The optimal solution is not strategyproof as the optimal solution with one obnoxious  UAV is not strategyproof as shown in Section \ref{sectionoflg}.
Our objective is to design a strategyproof mechanism with a small approximation ratio.
\begin{mechanism}\label{m10}
Given any location profile $(\textbf x,\textbf y)\in I^{2n},$ if $k$ is even, the locations of $k$ UAVs are $(X_j, Y_j, Z_j)=(0, 0, Z_j)$ with $j=1,2,\dots,\frac{k}{2}$, and $(X_j, Y_j, Z_j)=(2A,2B, Z_j)$ with $j=\frac{k}{2}+1, \dots, k;$ if $k$ is odd, the locations of $k$ UAVs are $(X_j, Y_j, Z_j)=(0,0,Z_j)$ with $j=1,2,\dots,\frac{k+1}{2}$, and $(X_j, Y_j, Z_j)=(2A,2B,Z_j)$ with $j=\frac{k+1}{2}+1, \dots, k.$
\end{mechanism}
\begin{theorem}\label{t10}
	Mechanism \ref{m10} is a stragegyproof mechanism with approximation ratio
	\[\gamma= \left\{ \begin{array}{ll}
	2, &\mbox{ if } k\mbox{ is even};\\
	\frac{2k}{k-1}, &\mbox{ if } k\mbox{ is odd}.\end{array} \right.\]
\end{theorem}
\begin{proof}
In $x$-, and $y$-domains,
the UAVs are located at the endpoints of each domain.
Mechanism \ref{m10} is strategyproof since the locations of UAVs are fixed and are independent of users' reports.
We next prove approximation ratio $\gamma.$ Due to the symmetry in every domain, we only need to consider $x$-domain. Define the number of UAVs deployed at $0$ in $x$-domain as $k_1$, and the number of UAVs deployed at $2A$ as $k_2$ with $k_1+k_2=k$. According to Mechanism \ref{m10}, $k_1=k_2=\frac{k}{2}$ if $k$ is even, and $k_1=\frac{k+1}{2}$ and $k_2=\frac{k-1}{2}$ if $k$ is odd.
Note that the optimal social utility is the maximum of
$SU_x=\sum_{j=1}^{k}\sum_{i=1}^{n}w_i (x_i-X_j)^2$,
by choosing $X_1, \cdots, X_k$ jointly, and the approximation ratio is
\begin{align}
	&\gamma=\frac{\max_{X_1,\dots,X_k}\sum_{j=1}^{k}\sum_{i=1}^{n} w_i(x_i-X_j)^2}{k_1\sum_{i=1}^{n} w_i x_i^2+k_2\sum_{i=1}^{n} w_i (x_i-2A)^2}.\nonumber
\end{align}
If weighted mean  $\bar{x}=\frac{\sum_{i=1}^n w_i x_i}{\sum_{i=1}^n w_i}\geq A,$ we have
\begin{align}
	\gamma=&\frac{k\sum_{i=1}^{n} w_i x_i^2}{k_1\sum_{i=1}^{n} w_i x_i^2+k_2\sum_{i=1}^{n} w_i (x_i-2A)^2}\nonumber\\
=&\frac{k}{k_1+k_2\frac{\sum_{i=1}^{n} w_i (x_i-2A)^2}{w_i x_i^2}}\leq \frac{k}{k_1}.\nonumber
\end{align}
If weighted mean  $\bar{x}=\frac{\sum_{i=1}^n w_i x_i}{\sum_{i=1}^n w_i}< A,$ we have
\begin{align}
	\gamma=&\frac{k\sum_{i=1}^{n} w_i (x_i-2A)^2}{k_1\sum_{i=1}^{n} w_i x_i^2+k_2\sum_{i=1}^{n} w_i (x_i-2A)^2}\nonumber\\
=&\frac{k}{k_1\frac{w_i x_i^2}{\sum_{i=1}^{n} w_i(x_i-2A)^2}+k_2}\leq \frac{k}{k_2}.\nonumber
\end{align}

Therefore, we have \[\gamma=\max\{\frac{k}{k_2}, \frac{k}{k_1}\}=\left\{ \begin{array}{ll}
	2, & \mbox{ if }k\mbox{ is even};\\
	\frac{2k}{k-1}, & \mbox{ if }k\mbox{ is odd}.\end{array} \right.\]
\end{proof}
On the other hand, besides the obnoxious multi-UAV placement game, we can also consider the multi-UAV placement game including  $k\geq 2$ of favorable UAVs, where each user of type 1 wants to minimize his cost to enjoy wireless service provided by the closest one out of $k$ UAVs. Thus, the service cost of user $i$ should be
\begin{align}
	c_i(f(\textbf x,\textbf y),(x_i,y_i))
=\min_{j\leq k}w_i ((x_i-X_j)^2+(y_i-Y_j)^2+Z_j^2).\nonumber
\end{align}
The optimal solution is not strategyproof as the optimal solution with one favorable UAV in Section \ref{sectionflg} is not strategyproof.
We want to design a strategyproof mechanism.

We propose the percentile mechanism for $k$ cooperative UAVs, which is inspired by the weighted median Mechanism \ref{m2}
and paper \cite{sui2013analysis}.
Denote $W=w_1+w_2+\dots+w_n$.
Define a new sequenced user number set $\Lambda_x$ in $x$-domain,
\begin{align}
\Lambda_x&=\{\lambda_{x,1},\lambda_{x,2},\dots,\lambda_{x,W}\}\nonumber\\
&=\{\underbrace{x_{1},\dots,x_{1}}_{w_{1}},\underbrace{x_{2},\dots,x_{2}}_{w_{2}},\dots,\underbrace{x_{n},\dots,x_{n}}_{w_{n}}\},\nonumber
\end{align}
where $\lambda_{x,1}$ satisfies that $\lambda_{x,i}\leq \lambda_{x,i+1},$ for any $i\geq 1.$
Basically, we rescale weights to be positive integers and partition each user $i$ of weight $w_i$ into $w_i$ users of unit weight.
Denote $p_j=\frac{j}{k+1}$ with $j=1,2,\dots, k$ as the percentile for $k$ UAVs.
Given a reported location profile $(\textbf x,\textbf y),$ our percentile mechanism locates $x$-location of the $j$-th UAV by selecting the $p_j$-percentile of the ordered projection of $\Lambda_x$ in the $x$-domain as location of UAV $j$ in $x$-domain. Formally,
\begin{align}
&(X_1,X_2,\dots,X_k)\nonumber\\
=&(\lambda_{x,\lfloor (W-1)p_1\rfloor+1}, \lambda_{x,\lfloor (W-1)p_2\rfloor+1}, \dots, \lambda_{x,\lfloor (W-1)p_k\rfloor+1}).\nonumber
\end{align}
The locations of the $j$-th UAV in $y$- and $z$-domians follow the same percentile strategy.
Actually, Mechanism \ref{m2} for $k=1$ UAV can be considered as a special case of percentile
mechanism with $p_1 = 1/2$ . We can also follow the strategyproof proof of Mechanism  \ref{m2} to prove that this percentile mechanism is strategyproof.

We illustrate multi-UAV placement  strategy by using an example of
14 users and 3 UAVs with each $w_i=1$. With $ (p_1, p_2, p_3) = (0.25, 0.5,0.75)$, the percentile
mechanism locates  $X_1$ for the first UAV to the $x$-coordinate of the  fourth ordered  user's location
(due to $\lfloor 0.25\times(14-1)\rfloor+1=4$),
$X_2$ for the second UAV to the $x$-coordinate of the seventh ordered user's location (due to $\lfloor 0.5\times(14-1)\rfloor+1=7$)
and $X_3$ for the last UAV to  the $x$-coordinate of the tenth ordered user' location (due to $\lfloor 0.75\times(14-1)\rfloor+1=10$).

\section{Conclusions}

We studied the algorithmic game theory problem to determine the final deployment location of a UAV in a 3D space, by learning selfish users' truthful locations and preferences.
To minimize the social cost in the UAV placement game, we designed the  strategyproof mechanism  with approximation ratio $2^{(3\alpha-4)/2}$, as compared to the social optimum with full information.
We also studied the obnoxious UAV placement game to maximize the social utility of a group of interfered users and proposed a strategyproof mechanism with approximation ratio $5\times 2^{(\alpha-2)/2}$.
Besides the worst-case analysis,
we proved that the empirical performances of the proposed
mechanisms improve with the number of users.
Moveover, we studied the dual-preference UAV placement game for the coexistence of the two groups of users,  and proposed a strategyproof mechanism with approximation ratio $2^{3\alpha/2}$.
Finally, We  extended the three placement games to include more than one UAV by designing strategyproof mechanisms and proving their approximation ratios.

In the future, we will study the dual-preference UAV placement game and consider user $i$ can misreport his weight $w_i.$
We will further consider randomized strategyproof mechanism designs in the three kinds of UAV placement games.




\begin{thebibliography}{10}

\bibitem{xu2018uav}
X.~Xu, L.~Duan, and M.~Li,
\newblock ``UAV placement games for optimal wireless service provision,"
\newblock {\em Proc. IEEE WiOpt}, pp. 1--8, 2018.

\bibitem{ernest2016genetic}
N.~Ernest, D.~Carroll, C.~Schumacher, M.~Clark, K.~Cohen, and G.~Lee,
\newblock ``Genetic fuzzy based artificial intelligence for unmanned combat
  aerial vehicle control in simulated air combat missions,"
\newblock {\em J Def Manag}, 6(144):2167--0374, 2016.

\bibitem{mozaffari2016efficient}
M.~Mozaffari, W.~Saad, M.~Bennis, and M.~Debbah,
\newblock ``Efficient deployment of multiple unmanned aerial vehicles for optimal
  wireless coverage,"
\newblock {\em IEEE Commun. Lett.}, 20(8):1647--1650, 2016.

\bibitem{drones-web}
A.~Pregler,
\newblock ``When cows fly: At\&t sending lte signals from drones,"
\newblock \url{http://about.att.com/innovationblog/cows_fly}, 2017.

\bibitem{pressman}
A.Pressman,
\newblock ``Verizon is testing drones for providing emergency cell service,"
\newblock \url{http://fortune.com/2016/10/06/verizon-drones-emergency/}, 2016.

\bibitem{6863654}
A.~Hourani, S.~Kandeepan, and S.~Lardner,
\newblock ``Optimal lap altitude for maximum coverage,"
\newblock {\em IEEE Wireless Commun.~Lett.}, 3(6):569--572, 2014.

\bibitem{7888557}
Y.~Zeng and R.~Zhang,
\newblock ``Energy-efficient uav communication with trajectory optimization,"
\newblock {\em IEEE Trans. Wireless Commun.}, 16(6):3747--3760, June 2017.

\bibitem{zhang2018fast}
X.~Zhang and L.~Duan,
\newblock ``Fast deployment of uav networks for optimal wireless coverage,"
\newblock {\em IEEE Trans. Mobile Comput.}, 2018.

\bibitem{pan2012cooperative}
M.~Pan, P.~Li, and Y.~Fang,
\newblock ``Cooperative communication aware link scheduling for cognitive
  vehicular networks,"
\newblock {\em IEEE JSAC},
  30(4):760--768, 2012.

\bibitem{wang2019dynamic}
X.~Wang and L.~Duan,
\newblock ``Dynamic pricing and capacity allocation of uav-provided mobile
  services,"
\newblock  {\em Proc. IEEE INFOCOM}, pp. 1--9, 2019.

\bibitem{901174}
G.~M.~Djuknic and R.~E.~Richton,
\newblock ``Geolocation and assisted gps,"
\newblock {\em IEEE Computer}, 34(2):123--125, Feb. 2001.

\bibitem{gu2009survey}
Y.~Gu, A.~Lo, and I.~Niemegeers,
\newblock ``A survey of indoor positioning systems for wireless personal
  networks,"
\newblock {\em IEEE Commun. Surveys Tuts.}, 11(1):13--32, 2009.

\bibitem{kos2006mobile}
T.~Kos, M.~Grgic, and G.~Sisul,
\newblock ``Mobile user positioning in gsm/umts cellular networks,"
\newblock {\em IEEE ELMAR}, pp. 185--188, 2006.

\bibitem{ProcacciaandTennenholtz}
A.~D.~Procaccia and M.~Tennenholtz,
\newblock ``Approximate mechanism design without money,"
\newblock {\em Proc. ACM EC}, pp. 177--186, 2009.


\bibitem{cheng2011mechanisms}
Y.~Cheng, W.~Yu, and G.~Zhang,
\newblock ``Mechanisms for obnoxious facility game on a path,"
\newblock {\em Proc. COCOA}, pp. 262--271, Springer, 2011.

\bibitem{ibara2012characterizing}
K.~Ibara and H.~Nagamochi,
\newblock ``Characterizing mechanisms in obnoxious facility game,"
\newblock {\em Proc. COCOA}, pp. 301--311, Springer, 2012.

\bibitem{zou2015facility}
S.~Zou and M.~Li,
\newblock ``Facility location games with dual preference,"
\newblock {\em Proc. AAMAS}, pp. 615--623, 2015.

\bibitem{feigenbaum2015strategyproof}
I.~Feigenbaum and J.~Sethuraman,
\newblock ``Strategyproof mechanisms for one-dimensional hybrid and obnoxious
  facility location models,"
\newblock In {\em Workshop on Incentive and Trust in E-Communities at AAAI}, 2015.

\bibitem{sui2013analysis}
X.~Sui, C.~Boutilier, and T.~W.~Sandholm,
\newblock ``Analysis and optimization of multi-dimensional percentile mechanisms,"
\newblock {\em Proc. AAAI}, pp. 367¡ª374, 2013.

\bibitem{wangzhe2018}
L.~Duan, Z.~Wang and R.~Zhang,
\newblock ``Adaptive deployment for uav-aided communication networks,"
\newblock {\em IEEE Trans. Wireless Commun.}, 2018.

\bibitem{paulraj2003introduction}
A.~Paulraj, R.~Nabar, and D.~Gore,
\newblock In {\em Introduction to space-time wireless communications},
\newblock Cambridge Univ.~Press, 2003.


\bibitem{hassani2000dirac}
S.~Hassani,
\newblock ``Dirac delta function,"
\newblock In {\em Mathematical Physics: A Modem Introduction to Its
  Foundations}, pp. 159--171, Springer, 1999.

\end{thebibliography}

\bibliographystyle{unsrt}

%

\begin{IEEEbiography}[{\includegraphics[width=1in,height=1.42in,clip,keepaspectratio]{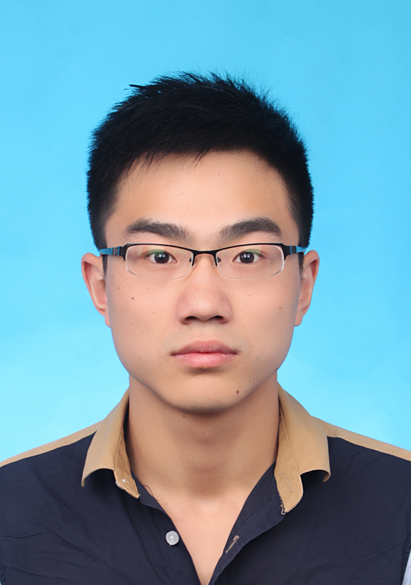}}]{Xinping Xu.pdf}
 (S'15) received the BS degree from the Department of Mathematics, Nanjing University, Nanjing, China, in 2015.
 He was a visiting research student in the Department of Computer Science, City University of Hong Kong for the period from 1 March 2019 to 30 April 2019.
 He is currently a PhD candidate in Engineering Systems and Design pillar at Singapore University of Technology and Design, Singapore.
 His research interests include algorithmic game theory, mechanism design, and network economics.

\end{IEEEbiography}

\begin{IEEEbiography}[{\includegraphics[width=1in,height=1.42in,clip,keepaspectratio]{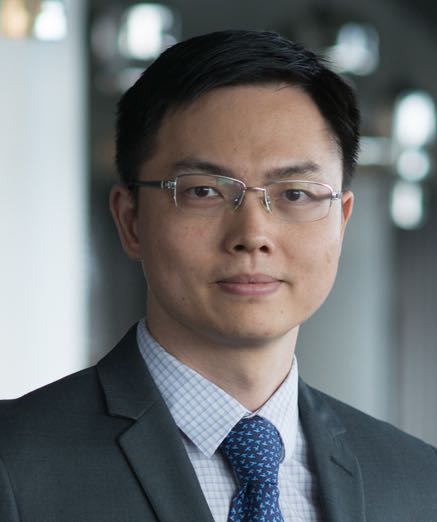}}]{Lingjie Duan.pdf}
(S'09-M'12-SM'17) received the Ph.D. degree from the Chinese University of Hong Kong in 2012. He is an Associate Professor of Engineering Systems and Design with the Singapore University of Technology and Design (SUTD). His research interests include network economics and game theory, cognitive communications and cooperative networking, and energy harvesting wireless communications. He is an Editor of IEEE Transactions on Wireless Communications and IEEE Communications Surveys and Tutorials.

\end{IEEEbiography}

\begin{IEEEbiography}[{\includegraphics[width=1in,height=1.42in,clip,keepaspectratio]{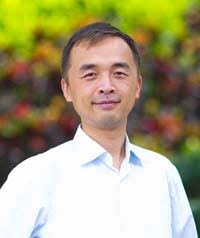}}]{Minming Li.pdf}
 (M'13-SM'15) received the BE and PhD degrees from the Department of Computer Science and Technology, Tsinghua University, Beijing, China, in 2002 and 2006, respectively. He is currently an associate professor in the Department of Computer Science, City University of Hong Kong, Hong Kong SAR. His research interests include algorithmic game theory, algorithm design and analysis, and combinatorial optimization.

\end{IEEEbiography}

\newpage

\begin{appendices}
\section{Lemma \ref{l1} and its proof}\label{app_1}
\begin{lemma}\label{l1}
If $C, D\geq 0$ and $E\geq 2,$ $(C+D)^E\geq C^E+D^E.$
\end{lemma}
\begin{proof}
\begin{align}\label{a53}
&(C+D)^E\geq C^E+D^E \Leftrightarrow
(1+\frac{D}{C})^E\geq 1+(\frac{D}{C})^E\nonumber\\
&\Leftrightarrow  \frac{(1+t)^E}{1+t^E}\geq 1 \mbox{ for } t=\frac{D}{C}.
\end{align}
To prove (\ref{a53}), we denote function $G(t)=\frac{(1+t)^E}{1+t^E}$ for $t\geq 0$ and $E\geq 2.$
We have
\begin{align}
&\frac{dG}{dt}=\frac{(1+t^E)E(1+t)^{E-1}-Et^{E-1}(1+t)^E}{(1+t^E)^2}=0\nonumber\\
&\Rightarrow t^{E-1}=1 \Rightarrow t=1.\nonumber
\end{align}
If $t\in[0,1)$ $\frac{dG}{dt}>0$; if $t\in(1, +\infty)$ $\frac{dG}{dt}<0.$
$G(t)$ obtain its minimum at $t=0$ or $t=+\infty.$
$G(t)\geq \min\{G(0), G(\infty)\}=1.$ Therefore (\ref{a53}) is proven.
\end{proof}

\section{Lemma \ref{l2} and its proof}\label{app_2}
\begin{lemma}\label{l2}
If $C, D\geq 0$ and $E\geq 2,$ $(C+D)^E\leq 2^{E-1}(C^E+D^E).$
\end{lemma}
\begin{proof}
\begin{align}
&(C+D)^E\leq 2^{E-1}(C^E+D^E)\nonumber\\
 \Leftrightarrow &
(1+\frac{D}{C})^E\leq2^{E-1}( 1+(\frac{D}{C})^E)\nonumber\\
\Leftrightarrow  & \frac{(1+t)^E}{1+t^E}\leq2^{E-1} \mbox{ for } t=\frac{D}{C}.\nonumber
\end{align}
By the proof of Lemma 1, $G(t)$ obtain its maximum at $t=1$ and $G(t)\leq G(1)=2^{E-1}.$
\end{proof}

\section{Proof of Theorem \ref{t3}}\label{app_3}
\begin{proof}	
We first look at Mechanism \ref{m1} and have
 \begin{align}\label{a44}
 \mathtt{Ratio.1}=\frac{SC_x+SC_y+\sum_{i=1}^n w_i z_0^2}{OPT_{1,x}+OPT_{1,y}+\sum_{i=1}^n w_i z_0^2}.
 \end{align}
 Due to symmetry, we only need to consider $x$-domain part of $\frac{SC_x}{OPT_{1,x}}$ in (\ref{a44}).
		For $SC_x$ in (\ref{a44}), we have	
\begin{align}
		SC_x&\!=\!\sum_{i=1}^{n}w_i(x_i-x_{med})^2\!=\!\sum_{i=1}^{n}w_i((x_i-\bar{x})\!+\!(\bar{x}-x_{med}))^2 \nonumber\\ &=(\bar{x}-x_{med})^2\sum_{i=1}^{n}w_i+\sum_{i=1}^{n}w_i(x_i-\bar{x})^2\nonumber\\
&=(\bar{x}-x_{med})^2\sum_{i=1}^{n}w_i+OPT_{1,x}.
		\label{a20}
		\end{align}
		We note $OPT_{1,x}$ in (\ref{a20}) does not converge to $0$, since
\begin{align} \E[{\sum_{i=1}^{n}w_i(x_i-\bar{x})^2}]
		\geq &nw_{min}\E[{(x_i-\E(x_1))^2}]\nonumber\\
=&nw_{min}\Var[x_1]>0. \nonumber
		\end{align}
		To prove $\mathtt{Ratio.1}$ in (\ref{a44}) converges in probability towards $1$,  we only need to prove the first term $(\bar{x}-x_{med})^2$ in (\ref{a20}) converges in probability to $0$ as $n$ goes to infinity, since
\begin{align}
&(\bar{x}-x_{med})^2\stackrel{P}{\rightarrow} 0 \nonumber\\
\Rightarrow & |\frac{SC_x}{OPT_{1,x}}-1|=\frac{(\bar{x}-x_{med})^2\sum_{i=1}^{n}w_i}{\sum_{i=1}^{n}w_i(x_i-\bar{x})^2}\nonumber\\
&\leq \frac{(\bar{x}-x_{med})^2nw_{max}}{nw_{min}\Var[x_1]}=\frac{w_{max}}{w_{min}\Var[x_1]}(\bar{x}-x_{med})^2\stackrel{P}{\rightarrow} 0\nonumber\\
\Rightarrow & \mathtt{Ratio.1}\stackrel{P}{\rightarrow} 1,\nonumber
\end{align}
		as $n\rightarrow \infty.$ Then we can show the followings.
		
		Assume $h$ and $H$ are the probability density function and cumulative distribution function of $x_i$, respectively, and the range of $x_i$ is $\eta$. Assume $g$ and $G$ are the probability density function and cumulative distribution function of $w_i$, respectively, and the range of $w_i$ is $\beta$.
		Assume all variables are continuous.	
		The expectation of random variable $\bar{x}$ is
		\begin{align}\label{a42}
		\E[\bar{x}]=&\int_{x_1\in\eta} \dots\int_{x_n\in\eta}\int_{w_1\in\beta} \dots\int_{w_n\in\beta} \nonumber\\ &\frac{\sum_{i=1}^{n}w_ix_i}{\sum_{i=1}^{n}w_i}(\prod_{j=1}^n h(x_j)g(w_j)) dw_1\dots dw_n dx_1 \dots dx_n\nonumber\\
		=&\int_{w_1\in\beta} \dots\int_{w_n\in\beta} \frac{\sum_{i=1}^{n}\E[x_1]w_i}{\sum_{i=1}^{n}w_i}g(w_1)\dots g(w_n)\nonumber\\
& dw_1\dots dw_n=\E[x_1].
		\end{align}
As $n\rightarrow \infty,$ due to
\begin{align}
&\Var[\frac{\sum_{i=1}^{n}w_ix_i}{n}]\rightarrow 0, \Var[\frac{n}{\sum_{i=1}^{n}w_i}]\rightarrow 0,\nonumber\\
&\E[\frac{\sum_{i=1}^{n}w_ix_i}{n}]\rightarrow \E[w_1]\E[x_1],\E[\frac{n}{\sum_{i=1}^{n}w_i}]\rightarrow \frac{1}{\E[w_1]},\nonumber
\end{align}
we have the variance of $\bar{x}$ is
\begin{equation}	\label{a43}	
\Var[\bar{x}] =\Var[\frac{\sum_{i=1}^{n}w_ix_i}{n}\frac{n}{\sum_{i=1}^{n}w_i}]\rightarrow 0.
		\end{equation}
		Therefore, due to (\ref{a42}) and (\ref{a43}), as $n\rightarrow \infty,$
$\bar{x}$ converges in probability towards $\E[x_1]$, i.e.,
\begin{equation}\label{a49}
\bar{x}\stackrel{P}{\rightarrow}\E[x_1].
\end{equation}
		
		On the other hand we consider random variable $x_{med}$ and we assume $n$ is odd. Otherwise, if $n$ is even we can get the same conclusion. From order statistics
, we have the distribution of $x_{med}$,
\[x_{med}(\zeta)=\frac{n!}{(\frac{n-1}{2})!(\frac{n-1}{2})!}h(\zeta)H(\zeta)^{\frac{n-1}{2}}(1-H(\zeta))^{\frac{n-1}{2}}, \zeta\in\eta.\]
		Thus we derive the expectation of $x_{med}$,
		\begin{align} \E[x_{med}]\!=&\!\!\!\int_{\zeta\in\eta}\!\!\!\zeta\frac{n!}{(\frac{n-1}{2})!(\frac{n-1}{2})!}h(\zeta)H(\zeta)^{\frac{n-1}{2}}(1\!-\!H(\zeta))^{\frac{n-1}{2}} d\zeta\nonumber\\
		=&\frac{n!}{((\frac{n-1}{2})!)^2}\int_{\zeta\in\eta}\zeta H(\zeta)^{\frac{n-1}{2}}(1-H(\zeta))^{\frac{n-1}{2}} dH(\zeta)\nonumber\\
		=&\frac{n!}{((\frac{n-1}{2})!)^2}\int_0^1 H^{-1}(\psi) \psi^{\frac{n-1}{2}}(1-\psi)^{\frac{n-1}{2}} d\psi\nonumber\\
=&\int_0^1 H^{-1}(\psi)S_n(\psi) d\psi,\nonumber
		\end{align}
		where $H^{-1}(\psi)$ for $0\leq\psi\leq 1$ is the inverse function of $H(\zeta)$ and $S_n(\psi)$ is a symmetric function satisfying \[S_n(\psi)=\frac{n!}{(\frac{n-1}{2})!(\frac{n-1}{2})!}\psi^{\frac{n-1}{2}}(1-\psi)^{\frac{n-1}{2}}.\]
$S_n(\psi)$ reaches its maximum at $\psi=\frac{1}{2}$ and by using Stirling's approximation
, its maximum is
\begin{align}	\label{a45} S_n(\frac{1}{2})&=\frac{n!}{(\frac{n-1}{2})!(\frac{n-1}{2})!}(\frac{1}{2})^{n-1}\approx \frac{\sqrt{2\pi n}(\frac{n}{e})^n}{2\pi \frac{n-1}{2}(\frac{\frac{n-1}{2}}{e})^{n-1}2^{n-1}}\nonumber\\
&=\frac{\sqrt{2\pi n}}{\pi e}(1+\frac{1}{n-1})^n\rightarrow\frac{\sqrt{2\pi n}}{\pi}\rightarrow \infty,
\end{align}
		as $n\rightarrow \infty$ and its value is $0$ when $\psi\in \{0,1\}.$ By using the property of Beta function 
, we have for any $n$,
\begin{equation}\label{a46}
\int_{\psi=0}^1 S_n(\psi) \,d\psi=\frac{n!}{(\frac{n-1}{2})!(\frac{n-1}{2})!}\frac{(\frac{n-1}{2})!(\frac{n-1}{2})!}{n!}=1.
\end{equation}
		Thus from (\ref{a45}) and (\ref{a46}), we can see that function $\lim\limits_{n\rightarrow \infty}S_n(\psi)$ of $\psi$ is Dirac delta function (see Chapter 6. Generalized Functions in \cite{hassani2000dirac}). Function $\lim\limits_{n\rightarrow \infty}S_n(\psi)$ actually is $\delta_{\frac{1}{2}}(\psi).$ By using the definite integral property of Dirac delta function (see Equation 6.4 in page 160 in \cite{hassani2000dirac}),
		we have	
		\begin{align}\label{a47}
		\lim\limits_{n\rightarrow \infty}\E[x_{med}]=&\lim\limits_{n\rightarrow \infty}\int_0^1 H^{-1}(\psi)S_n(\psi) d\psi\nonumber\\
		=&\int_0^1 H^{-1}(\psi)(\lim\limits_{n\rightarrow \infty}S_n(\psi)) d\psi\nonumber\\
=&\int_0^1 H^{-1}(\psi)\delta_{\frac{1}{2}}(\psi) d\psi
=H^{-1}(\frac{1}{2}).
		\end{align}	
		By using a similar method in (\ref{a47}), we derive as $n\rightarrow \infty$, the variance of $x_{med},$
\begin{align}\label{a48}
		\Var[x_{med}]
		&=\E[x_{med}^2]-(\E[x_{med}])^2\nonumber\\
		&=\int_{\psi=0}^1 (H^{-1}(\psi))^2 S_n(\psi) \,d\psi-(\E[x_{med}])^2\nonumber\\
		&\rightarrow  (H^{-1}(\frac{1}{2}))^2-(H^{-1}(\frac{1}{2}))^2=0.
		\end{align}
		Therefore, due to (\ref{a47}) and (\ref{a48}), as $n\rightarrow \infty,$
\begin{equation}\label{a50}
x_{med}\stackrel{P}{\rightarrow}H^{-1}(\frac{1}{2}).
\end{equation}

		Since all $x_i$ follow symmetric distribution, it holds that $\E[x_1]=H^{-1}(\frac{1}{2})$, which means the mean of $x_i$ is equal to the median of $x_i$.
	Due to (\ref{a49}), (\ref{a50}) and
		\[|\bar{x}-x_{med}|\leq |\bar{x}-\E[x_1]|+|x_{med}-H^{-1}(\frac{1}{2})|,\]
		we have proved  $|\bar{x}-x_{med}|\stackrel{P}{\rightarrow}0,$ as $n\rightarrow \infty$.
By checking (\ref{a20}), we have $\frac{SC_x}{OPT_{1,x}}\stackrel{P}{\rightarrow}1,$ as $n\rightarrow \infty.$ Therefore $\mathtt{Ratio.1}\stackrel{P}{\rightarrow}1,$
		Similarly, we can prove $\mathtt{Ratio.2}$ of Mechanism \ref{m2} converges in probability towards $1$, as $n\rightarrow \infty.$
\end{proof}

\section{Proof of Theorem \ref{t4}}\label{app_4}
\begin{proof}
We only consider $x$-domain location  to prove strategyproofness, as $y$-domain can be analyzed similarly.
Assume, without loss of generality, that $\sum_{x_i\in [0,A)}w_i\leq\sum_{x_i\in [A,2A]}w_i$. Thus, UAV's $x$-location of Mechanism \ref{m4} is $x=0$.
	We can see that any user in $[0,A)$ prefers  $x=2A$ and any user in $[A,2A]$ prefers $x=0.$ Any user in $[A,2A]$ is not willing to misreport his $x$-domain location, while any user in $[0,A)$  can not change  the relationship that $\sum_{x_i\in [0,A)}w_i\leq\sum_{x_i\in [A,2A]}w_i$ by misreporting his $x$-domain location.
	Thus, Mechanism \ref{m7}  is strategyproof. Next, we prove approximation ratio $\gamma.$

To obtain the approximation ratio, we need to let $z_0=0$ first.
Without loss of generality, assume that
\begin{align}
&\sum_{i:x_i\in [0,A)}\!w_i>\!\sum_{i:x_i\in [A,2A]}w_i,
\sum_{i:y_i\in [0,B)}\!w_i>\!\sum_{i:y_i\in [B,2B]}w_i.
\label{a10}
\end{align}
The other cases can be analyzed similarly.
Given condition (\ref{a10}), Mechanism \ref{m4} selects that $x=2A$ and $y=2B.$
Since the optimal solution can only select $x=0$ or $2A$ and $y=0$ or $2B$, to obtain the largest approximation ratio, the optimal solution must be $x_{opt}=0$ and $y_{opt}=0.$
By Lemma 2,
the optimal social utility  can be described as
\begin{align}
&OPT_2(\textbf x, \textbf y)=\sum_{i=1}^n w_i (x_i^2+y_i^2)^{\alpha/2}\nonumber\\
\leq &
\sum_{i=1}^n w_i 2^{\alpha/2-1}(x_i^\alpha+y_i^\alpha)\nonumber\\
=&2^{\alpha/2-1}(\sum_{i:x_i\in [0,A)} w_i x_i^\alpha+\sum_{i:x_i\in (A,2A]} w_i x_i^\alpha)\nonumber\\
&2^{\alpha/2-1}(\sum_{i:y_i\in [0,B)} w_i y_i^\alpha+\sum_{i:y_i\in (B,2B]} w_i y_i^\alpha)\nonumber\\
\leq & 2^{\alpha/2-1}(\sum_{i:x_i\in [0,A)} w_i A^\alpha+\sum_{i:x_i\in [0,A)} w_i (2A)^\alpha)\nonumber\\
 &+ 2^{\alpha/2-1}(\sum_{i:y_i\in [0,B)} w_i B^\alpha+\sum_{i:x_i\in [0,B)} w_i (2B)^\alpha)\nonumber\\
 =& 2^{\alpha/2-1} (5A^\alpha\sum_{i:x_i\in [0,A)} w_i +5B^\alpha\sum_{i:y_i\in [0,B)} w_i)
\label{a11}
\end{align}
Under Mechanism \ref{m4}, by Lemma 1, the social utility at $x=2A$ and $y=2B$ is
\begin{align}
&SU(f,(\textbf{x},\textbf y))
=\sum_{i=1}^n w_i ((2A-x_i)^2+(2B-y_i)^2)^{\alpha/2}\nonumber\\
\geq&  \sum_{i=1}^n w_i ((2A-x_i)^\alpha+(2B-y_i)^\alpha)\nonumber\\
=& \sum_{i: x_i\in [0,A)} w_i (2A-x_i)^\alpha+ \sum_{i: x_i\in [A,2A])} w_i (2A-x_i)^\alpha\nonumber\\
&+\sum_{i: y_i\in [0,B)} w_i (2B-y_i)^\alpha+ \sum_{i: y_i\in [B,2B])} w_i (2B-y_i)^\alpha\nonumber\\
\geq & 0+\sum_{i: x_i\in [0,A))} w_i A^\alpha+0+\sum_{i: y_i\in [0,B))} w_i B^\alpha.
\label{a12}
\end{align}
To determine the maximum approximation ratio $\gamma$ for the worst-case, we want to increase the optimal social utility in (\ref{a11}) and reduce the social utility of Mechanism \ref{m4} in (\ref{a12}). We  set $x_i=A$ for all $x_i\in [0,A)$ and $x_i=2A$ for all $x_i\in [A,2A]$.
by comparing  (\ref{a11}) and (\ref{a12}), we have
$OPT_{2,x}(\textbf{x}) \leq 2^{(\alpha-2)/2} 5SU_x(f,\textbf{x}).$
Hence, we conclude $\gamma\leq5\times 2^{(\alpha-2)/2}$ for Mechanism \ref{m4}.
\end{proof}
\section{Proof of Theorem \ref{t8}}\label{app_5}
\begin{proof}
We first prove the strategyproofness.
	We only consider $x$-domain, as $y$-domains can be analyzed similarly.
	Assume, without loss of generality, that $|R_{1}|+|L_{2}|<|R_{2}|+|L_{1}|$. UAV location of Mechanism \ref{m7} is $x=0$.
	We can see that any user in $R_1\cup L_2$ prefers  $x=2A$ and any user in $R_2\cup L_1$ prefers $x=0.$ Any user in $R_{2}\cup \ L_{1}$ has no incentive to misreport his $x$-domain location and preference type, and any user in $R_{1}\cup L_{2}$  can not change  the relationship   that $|R_{1}|+|L_{2}|<|R_{2}|+|L_{1}|$ by misreporting his $x$-domain location or preference type.
	Thus, Mechanism \ref{m7}  is strategyproof. Next, we prove the approximation ratio.
	
Assume, without loss of generality, that $|R_{1}|+|L_{2}|<|R_{2}|+|L_{1}|$ and $|\bar R_{1}|+|\bar L_{2}|<|\bar R_{2}|+|\bar L_{1}|$. UAV location of Mechanism \ref{m7} is $x=0$ and $y=0$.
To obtain the approximation ratio, we should let $z_0=0$ first.
By Lemma 2, the optimal utility can be described as
\begin{align}\label{a54}
&OPT_3(\textbf x, \textbf y)\nonumber\\
=&\max_{x,y} (\sum_{i:\theta_i=1}((2A)^2-(x_i-x)^2+(2B)^2-(y_i-y)^2)^{\alpha/2}\nonumber\\
&+\sum_{i:\theta_i=2} ((x_i-x)^2+(y_i-y)^2)^{\alpha/2})\nonumber\\
\leq & 2^{\frac{\alpha-2}{2}}\max_{x,y} (\sum_{i:\theta_i=1}((2A)^2-(x_i-x)^2)^{\frac{\alpha}{2}}+((2B)^2\nonumber\\
&-(y_i-y)^2)^{\frac{\alpha}{2}})
+\sum_{i:\theta_i=2} ((x_i-x)^\alpha+(y_i-y)^\alpha))\nonumber\\
=& OPT_{3,x}+OPT_{3,y},
\end{align}
where $OPT_{3,x}=2^{\frac{\alpha-2}{2}}\max_x(\sum_{i:\theta_i=1}((2A)^2-(x_i-x)^2)^{\frac{\alpha}{2}}+\sum_{i:\theta_i=2}(x_i-x)^\alpha)$
and
$OPT_{3,y}=2^{\frac{\alpha-2}{2}}\max_y(\sum_{i:\theta_i=1}((2B)^2-(y_i-y)^2)^{\frac{\alpha}{2}}+\sum_{i:\theta_i=2}(y_i-y)^\alpha).$
By Lemma 1, the social utility in  Mechanism \ref{m7} is
\begin{align}\label{a55}
&SU(f,(\textbf x, \textbf y))\nonumber\\
=&\sum_{i:\theta_i=1}((2A)^2-x_i^2+(2B)^2-y_i^2)^{\alpha/2}
+\sum_{i:\theta_i=2} (x_i^2+y_i^2)^{\alpha/2}\nonumber\\
\geq & \sum_{i:\theta_i=1}(((2A)^2-x_i^2)^{\frac{\alpha}{2}}+((2B)^2-y_i^2)^{\frac{\alpha}{2}})
+\sum_{i:\theta_i=2} (x_i^\alpha+y_i^\alpha)\nonumber\\
=& SU_x+SU_y,
\end{align}
where $SU_x=\sum_{i:\theta_i=1}((2A)^2-x_i^2)^{\frac{\alpha}{2}}+\sum_{i:\theta_i=2}x_i^\alpha$
and
$SU_y=\sum_{i:\theta_i=1}((2B)^2-y_i^2)^{\frac{\alpha}{2}}+\sum_{i:\theta_i=2}y_i^\alpha.$
We consider the case that $x\in [0,A]$ for the optimal UAV location.
	By checking the value ranges of $x_i$ and $x$,  $OPT_{3,x}$ satisfies  that
\begin{align}\label{a56}
&2^{-\frac{\alpha-2}{2}}OPT_{3,x}
=\max_x(\sum_{i:i\in R_1\cup L_1}((2A)^2-(x_i-x)^2)^{\frac{\alpha}{2}}\nonumber\\
&+\sum_{i:i\in R_2\cup L_2}(x_i-x)^\alpha)\nonumber\\
\leq &
\sum_{i:i\in R_2}(2A)^\alpha+\sum_{i:i\in L_2}A^\alpha+\sum_{i:i\in L_1}(2A)^\alpha+\sum_{i:i\in R_1}(2A)^\alpha\nonumber\\		=&|R_2|(2A)^\alpha+|L_2|A^\alpha+|L_1|(2A)^\alpha+|R_1|(2A)^\alpha\nonumber\\
		=&(|R_2|(2A)^\alpha+|L_1|(2A)^\alpha)+(|L_2|A^\alpha+|R_1|(2A)^\alpha)\nonumber\\
		\leq&(|R_2|(2A)^\alpha+|L_1|(2A)^\alpha)+(|L_2|(2A)^\alpha+|R_1|(2A)^\alpha)\nonumber\\
	\leq &(|R_2|(2A)^\alpha+|L_1|(2A)^\alpha)+(|R_2|(2A)^\alpha+|L_1|(2A)^\alpha)\nonumber\\
		=&2|R_2|(2A)^\alpha+2|L_1|(2A)^\alpha.
\end{align}
Then we consider the other case that $x\in (A,2A]$.
	By checking the value ranges of $x_i$ and $x$,  $OPT_{3,x}$ satisfies that
\begin{align}\label{a57}
&2^{-\frac{\alpha-2}{2}}OPT_{3,x}\nonumber\\
\leq &
\sum_{i:i\in R_2}A^\alpha+\sum_{i:i\in L_2}(2A)^\alpha+\sum_{i:i\in L_1}(2A)^\alpha+\sum_{i:i\in R_1}(2A)^\alpha\nonumber\\		=&|R_2|A^\alpha+|L_2|(2A)^\alpha+|L_1|(2A)^\alpha+|R_1|(2A)^\alpha\nonumber\\
		=&(|R_2|A^\alpha+|L_1|(2A)^\alpha)+(|L_2|(2A)^\alpha+|R_1|(2A)^\alpha)\nonumber\\
		\leq&(|R_2|A^\alpha+|L_1|(2A)^\alpha)+(|L_1|(2A)^\alpha+|R_2|(2A)^\alpha)\nonumber\\
=&|R_2|(A^\alpha+(2A)^\alpha)+ 2|L_1|(2A)^\alpha.
\end{align}
For $SU_x,$ we have
\begin{align}\label{a58}
SU_x=&\sum_{i:i\in R_1\cup L_1}((2A)^2-x_i^2)^{\frac{\alpha}{2}}+\sum_{i:i\in R_2\cup L_2}x_i^\alpha\nonumber\\
\geq&
\sum_{i\in R_2}A^\alpha+0+\sum_{i\in L_1}(4A^2-A^2)^{\frac{\alpha}{2}}+0\nonumber\\
		=&|R_2|A^\alpha+3^{\frac{\alpha}{2}}|L_1|A^\alpha.
\end{align}
Therefore, the approximation ratio is
\begin{align}
\gamma=&\frac{OPT_3(\textbf x, \textbf y)}{SU(f,(\textbf x, \textbf y))}\leq \frac{OPT_{3,x}+OPT_{3,y}}{SU_x+SU_y}= \frac{OPT_{3,x}}{SU_x}\nonumber\\
\leq&2^{\frac{\alpha-2}{2}}\frac{2|R_2|(2A)^\alpha+2|L_1|(2A)^\alpha}{|R_2|A^\alpha+3^{\frac{\alpha}{2}}|L_1|A^\alpha}\leq 2^{3\alpha/2},
\end{align}
where the first inequality is due to (\ref{a54}) and (\ref{a55}), and the second inequality is due to (\ref{a56}), (\ref{a57}) and (\ref{a58}).
\end{proof}

\section{Proof of Theorem \ref{t110}}\label{app_6}
\begin{proof}

First, we prove Mechanism \ref{m110} is strategyproof.
In Mechanism \ref{m110}, there are only two choices for the two
UAVs' $x$-locations: $(0,2A)$ and $(2A,0).$ By comparing the utility of user $i$ in  eight different sets $Q_1, \cdots, Q_8$ for points $(0,2A)$ and $(2A,0)$, we can obtain user $i$'s preference towards candidates $(0,2A)$ and $(2A,0)$.
Any user $i$ in $Q_2\cup Q_7$ prefers the locations of two UAVs to be $(0,2A)$ and any user $i$ in $Q_3\cup Q_6$ prefers the locations of two UAVs to be $(2A,0).$
Any user $i$ in $Q_1\cup Q_4\cup Q_5 \cup Q_8$ are indifferent.

Assume, without loss of generality, that $|Q_2|+|Q_7|\geq |Q_3|+|Q_6|$. $(X_1,X_2)$ in Mechanism \ref{m110} should be $(0,2A)$.
Any user $i$ in $Q_2\cup Q_7$ is not willing to misreport his location $x_i$ or his preference $(\theta_i^1,\theta_i^2)$, since $(X_1,X_2)=(0,2A)$ is already the best choice for him.
Misreporting location $x_i$ and preference $(\theta_i^1,\theta_i^2)$ by any user $i$ in $Q_3\cup Q_6$ does not change the fact that $|Q_2|+|Q_7|\geq |Q_3|+|Q_6|$ and $(X_1,X_2)$ is still $(0,2A).$
Any user $i$ in $Q_1\cup Q_4\cup Q_5 \cup Q_8$ has no incentive to change his location $x_i$ or his preference $(\theta_i^1,\theta_i^2)$, since $(X_1,X_2)=(0,2A)$ or  $(X_1,X_2)=(2A,0)$ are the same for him.

For the other case $|Q_2|+|Q_7|< |Q_3|+|Q_6|$, the same conclusion can be drawn.
Therefore, misreporting location $x_i$ and preference $(\theta_i^1,\theta_i^2)$ by any user $i$ does not increase his utility $u_i$ and Mechanism \ref{m110} is strategyproof in the $x$-domain. The strategyproof result can be similarly proved for the $y$-domain and $z$-domain.

Next, we prove approximation ratio $\gamma.$ Without loss of generality, assume that $|Q_2|+|Q_7|\geq |Q_3|+|Q_6|$ and thus $(X_1,X_2)=(0,2A)$ for $x$-domain in this mechanism.
For the optimal utility, due to $|Q_2|+|Q_7|\geq |Q_3|+|Q_6|$,
we have
\begin{align}\label{a40}
& OPT_{3,x}=\max\limits_{(X_1,X_2)}\sum_{i\in N}((X_1-x_i)^2+(X_2-x_i)^2)\nonumber\\
\leq & n\times ((2A)^2+(2A)^2)=8A^2n=8A^2\sum_{i=1}^8|Q_i|\nonumber\\
\leq & 8A^2(|Q_1|+|Q_4|+|Q_5|+|Q_8|)+16A^2(|Q_2|+|Q_7|).
\end{align}
For the social utility of Mechanism \ref{m110}, by (\ref{a131}), we have
\begin{align}\label{a41}
&SU_x(0,1)=\sum_{i\in Q_4\cup Q_8}(x_i^2+(2A-x_i)^2)\nonumber\\
&+\sum_{i\in Q_1\cup Q_5}(8A^2-x_i^2-(2A-x_i)^2)\nonumber\\
&+\sum_{i\in Q_2\cup Q_6}((2A)^2-x_i^2+(2A-x_i)^2)\nonumber\\
&+\sum_{i\in Q_3\cup Q_7}(x_i^2+(2A)^2-(2A-x_i)^2)\nonumber\\
\geq &2A^2 (|Q_4|+|Q_8|)+4A^2 (|Q_1|+|Q_5|)+(2A)^2|Q_2|\nonumber\\
&+0\times|Q_6|+0\times|Q_3|+(2A)^2\times|Q_7|\nonumber\\
\geq &2A^2 (|Q_4|+|Q_8|+|Q_1|+|Q_5|)+4A^2(|Q_2|+|Q_7|).
\end{align}
By (\ref{a40}) and (\ref{a41}), we have the approximation ratio,
$\gamma=\frac{OPT_{3,x}}{SU_x(0,2A)}\leq 4.$
\end{proof}

\end{appendices}

\end{document}